\newif\ifsubmission
\submissionfalse

\ifsubmission
\documentclass{llncs}
\else
\documentclass{article}
\usepackage{palatino}
\usepackage{fullpage}
\fi

\usepackage{amsthm}

\usepackage{amsmath,amssymb,enumerate,algorithm,algorithmic,latexsym}
\usepackage{tikz}
\usetikzlibrary{fpu}
\usepackage{footmisc}
\usepackage{breakcites}
\usepackage{float}
\newfloat{algorithm}{H}{lop}
\usepackage{xspace}
\usepackage{amsfonts}
\usepackage{caption}
\usepackage{enumitem}
\usepackage[normalem]{ulem}
\usepackage{url}
\usepackage{mathtools}
\usepackage[capitalize]{cleveref}
\usepackage[
	lambda,
	operators,
	advantage,
	sets,
	adversary,
	landau,
	probability,
	notions,
	logic,
	ff,
	mm,
	primitives,
	events,
	complexity,
	asymptotics,
	keys]{cryptocode}
\usepackage{braket}
\usepackage{authblk}

\ifsubmission
\renewcommand{\paragraph}[1]{\medskip\noindent\textbf{#1}}
\else
 \newtheorem{theorem}{Theorem}[section]
 \newtheorem{definition}[theorem]{Definition}
 \newtheorem{remark}[theorem]{Remark}
 \newtheorem{lemma}[theorem]{Lemma}

\fi

\newcommand{\Eval}{\mathsf{Eval}}
\newcommand{\PRF}{\mathsf{PRF}}
\newcommand{\PRG}{\mathsf{PRG}}
\newcommand{\Puncture}{\mathsf{Punct}}
\newcommand{\NIZK}{\mathsf{NIZK}}

\newcommand{\WE}{\mathsf{WE}}
\newcommand{\ABE}{\mathsf{ABE}}
\newcommand{\OWF}{\mathsf{OWF}}
\newcommand{\ZAP}{\mathsf{ZAP}}
\newcommand{\ZAPR}{\mathsf{ZAPR}}
\newcommand{\NIWI}{\mathsf{NIWI}}
\newcommand{\SBSH}{\mathsf{SBSH}}
\newcommand{\KeyGen}{\mathsf{KeyGen}}
\newcommand{\TKeyGen}{\mathsf{TdGen}}
\newcommand{\TVerify}{\mathsf{TdVerify}}
\newcommand{\SKeyGen}{\mathsf{SimGen}}
\newcommand{\simgen}{\mathsf{sim}}
\newcommand{\Key}{\mathsf{Key}}
\newcommand{\Enc}{\mathsf{Enc}}
\newcommand{\Dec}{\mathsf{Dec}}
\newcommand{\Gen}{\mathsf{Gen}}
\newcommand{\Ext}{\mathsf{Ext}}
\newcommand{\Setup}{\mathsf{Setup}}
\newcommand{\Prove}{\mathsf{Prove}}
\newcommand{\Verify}{\mathsf{Verify}}

\newcommand{\cPRF}{\mathsf{cPRF}}
\newcommand{\Constrain}{\mathsf{Constrain}}
\newcommand{\CEval}{\mathsf{CEval}}

\newcommand{\share}{\mathsf{Share}}
\newcommand{\reconstruct}{\mathsf{Rec}}

\newcommand{\QFHE}{\mathsf{QFHE}}

\newcommand{\mpk}{\mathsf{mpk}}
\newcommand{\msk}{\mathsf{msk}}

\newcommand{\yes}{\mathsf{yes}}
\newcommand{\no}{\mathsf{no}}

\newcommand{\obfC}{\widetilde{\mathbf{CC}}}
\renewcommand{\CC}[3]{\ensuremath{\mathbf{CC}[#1,#2,#3]}}
\newcommand{\ccSim}{\mathsf{Sim}}
\renewcommand{\obf}{\mathsf{Obf}}

\newcommand{\crs}{\mathsf{crs}}
\newcommand{\td}{\mathsf{td}}
\newcommand{\ck}{\mathsf{ck}}
\newcommand{\params}{\mathsf{pp}}
\newcommand{\secp}{\mathsf{\lambda}}
\newcommand{\ct}{\mathsf{ct}}
\newcommand{\Hyb}{\mathsf{Hyb}}

\newcommand{\com}{\mathsf{Com}}

\newcommand{\QMA}{\mathsf{QMA}}
\newcommand{\mQMA}{\mathsf{mQMA}}
\newcommand{\NP}{\mathsf{NP}}
\newcommand{\Sim}{\mathsf{Sim}}
\newcommand{\lang}{\mathcal{L}}
\newcommand{\relation}{\mathsf{R}_\mathcal{L}}

\newcommand{\cL}{\mathcal{L}}
\newcommand{\cH}{\mathcal{H}}
\newcommand{\cA}{\mathcal{A}}
\newcommand{\cO}{\mathcal{O}}
\newcommand{\cV}{\mathcal{V}}

\newcommand{\cG}{\mathcal{G}}
\newcommand{\cI}{\mathcal{I}}

\newcommand{\bbN}{\mathbb{N}}

\newenvironment{boxfig}[2]{\begin{figure}[#1]\fbox{\begin{minipage}{\linewidth}
                        \vspace{0.2em}
                        \makebox[0.025\linewidth]{}
                        \begin{minipage}{0.95\linewidth}
            {{
                        #2 }}
                        \end{minipage}
                        \vspace{0.2em}
                        \end{minipage}}}{\end{figure}}

\newcommand{\pprotocol}[4]{
\begin{boxfig}{h!}{
\begin{center}
\textbf{#1}
\end{center}
    #4
\vspace{0.2em} } \caption{\label{#3} #2}
\end{boxfig}
}
\newcommand{\protocol}[4]{
\pprotocol{#1}{#2}{#3}{#4} }

\newcommand{\authnote}[3]{\textcolor{#3}{[{\footnotesize {\bf #1:} { {#2}}}]}}
\newcommand{\gnote}[1]{\authnote{Giulio}{#1}{red}}
\newcommand{\jnote}[1]{\authnote{James}{#1}{blue}}

\begin{document}

\ifsubmission
\title{Indistinguishability Obfuscation of Null Quantum Circuits and Applications}
\author{}
\institute{}
\else
\title{Indistinguishability Obfuscation of Null Quantum Circuits and Applications}
\author[1]{James Bartusek}
\author[2]{Giulio Malavolta}
\affil[1]{University of California, Berkeley}
\affil[2]{Max Planck Institute for Security and Privacy}
\date{}
\fi

\maketitle

\begin{abstract}
    We study the notion of indistinguishability obfuscation for null quantum circuits (quantum null-iO). We present a construction assuming:
    \begin{itemize}
        \item The quantum hardness of learning with errors (LWE).
        \item Post-quantum indistinguishability obfuscation for \emph{classical} circuits.
        \item A notion of ``dual-mode'' classical verification of quantum computation (CVQC).
    \end{itemize}
    We give evidence that our notion of dual-mode CVQC exists by proposing a scheme that is secure assuming LWE in the quantum random oracle model (QROM). 
    
   Then we show how quantum null-iO enables a series of new cryptographic primitives that, prior to our work, were unknown to exist even making heuristic assumptions. Among others, we obtain the first witness encryption scheme for QMA, the first publicly verifiable non-interactive zero-knowledge (NIZK) scheme for QMA, and the first attribute-based encryption (ABE) scheme for BQP.
\end{abstract}
\section{Introduction}

The goal of program obfuscation~\cite{AC:Hada00,C:BGIRSVY01} is to convert an arbitrary circuit $C$ into an unintelligible but functionally equivalent circuit $\widetilde{C}$. Recent work has shown that program obfuscation enables a series of new remarkable applications (e.g.~\cite{FOCS:GGHRSW13,STOC:SahWat14,EC:GGGJKL14,FOCS:BitPanRos15}), establishing obfuscation as a central object in cryptography. 

Yet, the scope of obfuscation has so far been concerned almost exclusively with classical cryptography. The advent of quantum computing has motivated researchers~\cite{aaronson,qobf} to ask whether program obfuscation is a meaningful notion also in a quantum world:
\begin{quote}\centering
\emph{Can we obfuscate quantum circuits? Is this notion useful at all?} 
\end{quote}
 Unfortunately, results on the matter are largely negative~\cite{qobf,ssl,impobf}, barring a few schemes for restricted function classes of questionable usefulness~\cite{braidobf,lowTobf}. At present, it is unclear whether obfuscation of quantum circuits in its most general form can exist at all. The goal of this work is to make progress on this question.

\subsection{Our Results}

In this work, we study the notion of obfuscation for quantum circuits. Our contributions are twofold.

\subsubsection{Quantum Null-iO and Witness Encryption for QMA} 

We show that, assuming LWE, post-quantum indistinguishability obfuscation (iO) for classical circuits, and (trapdoor) dual-mode classical verification of quantum computation (CVQC), there exists an obfuscation scheme for null quantum circuits, i.e., any polynomial-size quantum circuit that rejects all inputs with overwhelming probability. The following statement summarizes our main result.

\begin{theorem}[Informal]
Assuming the hardness of the LWE problem, the existence of post-quantum iO for classical circuits, and a (trapdoor) dual-mode CVQC protocol, there exists quantum null-iO.
\end{theorem}
While the first assumption is standard, and the second seems to some extent necessary, dual-mode CVQC is a non-standard cryptographic building block that we introduce in this work. Loosely speaking, a CVQC protocol is dual-mode if there is a standard mode in which the scheme is correct, and a simulation mode in which there do not exist \emph{any} accepting proofs for no instances (though in this mode, the scheme may not necessarily be correct for yes instances). These modes must be computationally indistinguishable even given the verification key. 

Actually, we do not know of any construction of CVQC that satisfies this dual-mode property, so we instead relax the property to a ``trapdoor'' variant, where there exists a trapdoor setup algorithm  (computationally indistinguishable from the original one) that satisfies the dual-mode property. We show that this relaxation suffices to construct quantum null-iO (along with LWE and post-quantum iO for classical circuits), and present a construction of trapdoor dual-mode CVQC secure against the learning with errors (LWE) problem in the quantum random oracle model (QROM).

\begin{theorem}[Informal]
Assuming the hardness of the LWE problem, there exists a trapdoor dual-mode CVQC protocol in the QROM.
\end{theorem}

In addition, we propose an alternative construction of quantum null-iO (Appendix~\ref{sec:vbb}) secure only against the LWE problem but with respect to an oracle. The  oracle that we consider is entirely classical, but queriable in superposition. In fact, we show that the scheme is secure assuming LWE and post-quantum virtual black-box obfuscation of a particular classical circuit.

\paragraph{Witness Encryption for QMA.}
Applying a well-known transformation, we obtain witness encryption \cite{STOC:GGSW13} for QMA as a corollary. Importantly, our scheme has an entirely classical encryption algorithm: Any classical user can encrypt a message $m$ with respect to the membership of some statement $x$ in a language $\lang \in \QMA$. The message $m$ can be (quantumly) decrypted by anyone possessing (multiple copies of) a valid witness $\ket{\psi}\in\relation(x)$.

\subsubsection{New Applications}

We show that witness encryption for QMA with classical encryption enables a series of new cryptographic primitives, thereby positioning witness encryption for QMA (and consequently quantum null-iO) as a central catalyst in quantum cryptography. Most of our results are obtained via \emph{classical synthesis of quantum programs}: We compress an exponential number of quantum programs into a small classical circuit via the use of classical iO. We give an overview of the implications of our results below. We remark that, prior to our work, we did not even have a heuristic candidate for \emph{any} of the primitives that we obtain.

\begin{itemize}
    \item[(1)] \textbf{NIZK for QMA:} We present the first construction of NIZK~\cite{STOC:BluFelMic88} for QMA. A (quantum) prover can efficiently produce a zero-knowledge certificate $\pi$ that a certain statement $x\in\lang$, where $\lang$ is any language in QMA. This certificate is publicly verifiable with respect to a publicly-known common reference string (CRS). Prior to our work, all non-interactive proof systems for QMA~\cite{BG19,C:ColVidZha20,TCC:ACGH20,TCC:ChiChuYam20} were in the \emph{secret parameters model} or the \emph{designated verifier} setting, i.e.\ the verifier (or additionally the prover) needed some secret information not accessible to the other party.
    
    This resolves an outstanding open problem in the area (see e.g.~\cite{Shm20} for a discussion on the barriers to achieving public verifiability). In addition, our NIZK scheme satisfies several properties of interest, namely, (i) it is statistically zero-knowledge, (ii) the verification algorithm (and CRS) is fully classical, and (iii) the verification algorithm is \emph{succinct}, i.e.\ its runtime is independent of the size of the witness. In fact, we obtain the first succinct non-interactive argument (zk-SNARG) for QMA. 
    
    This primitive also implies the first classical verification of quantum computation (CVQC) scheme with succinct verifier (and large CRS) that is \emph{publicly verifiable}, and thus reusable. This improves over the privately-verifiable scheme of \cite{TCC:ChiChuYam20} where the large CRS is not reusable, though we note that their CRS is actually a common \emph{random} string as opposed to the common reference string required for our protocol.
    
    \item[(2)] \textbf{ZAPR for QMA:} We show how to transform our NIZK for QMA scheme into a (publicly verifiable) two-round statistically witness-indistinguishable argument (ZAPR) for QMA. Our transformation is generic and can be thought of as a quantum analogue of the Dwork-Naor compiler~\cite{FOCS:DwoNao00}, in the setting of computational soundness.
    
    \item[(3)] \textbf{ABE for BQP:} We obtain a ciphertext-policy ABE~\cite{EC:SahWat05,CCS:GPSW06} scheme for BQP (bounded-error quantum polynomial-time) computation. In ciphertext-policy ABE for BQP, anyone can encrypt a message $m$ with respect to some BQP language, represented by a quantum circuit $Q$. The key authority can generate decryption keys associated with any attribute $x$. The ciphertext can then be decrypted if and only if evaluating $Q$ on $x$ produces 1 (which, by QMA amplification, can happen with either overwhelming probability or negligible probability). Interestingly, all algorithms except for the decryption circuit are fully classical. This is the first example of an ABE scheme for functionalities beyond classical computations. 
    
    The scheme satisfies the standard notion of payload-hiding. That is, the message $m$ is hidden, though the policy $Q$ is revealed by the ciphertext. We then show that we can upgrade the security of the scheme via a generic transformation to predicate-encryption security~\cite{C:GorVaiWee15} (i.e.\ the policy $Q$ is hidden from the evaluator if they are only in possession of keys for rejecting attributes). We achieve this via a construction of lockable obfuscation~\cite{FOCS:GoyKopWat17,FOCS:WicZir17} for quantum circuits from LWE.

    \item[(4)] \textbf{Constrained PRF for BQP:} We present a construction of a pseudorandom function (PRF)~\cite{AC:BonWat13,CCS:KPTZ13,PKC:BoyGolIva14} where one can issue \emph{constrained} keys associated to a quantum circuit $Q$. Such keys can evaluate the PRF on an input $x$ if and only if evaluating $Q$ on $x$ returns 1 with overwhelming probability. Otherwise, the output of the PRF on $x$ looks pseudorandom. The scheme is fully collusion-resistant, i.e.\ security is preserved even if an unbounded number of constrained keys is issued.
    
    \item[(5)] \textbf{Secret Sharing for Monotone QMA:} Finally, as a direct application of our witness encryption scheme, we show how to construct a secret sharing scheme for access structures in monotone QMA.
\end{itemize}

\subsection{Technical Overview}

We give a cursory overview of the techniques introduced by our work and we provide some informal intuition on how we achieve our results.

\subsubsection{How to Obfuscate Quantum Circuits}

Before delving into the specifics of our approach, we highlight a few reasons why known techniques for obfuscating classical circuits do not seem to be directly portable to the quantum setting. For obvious reasons, we restrict this discussion to schemes that plausibly retain security in the presence of quantum adversaries. Recent proposals~\cite{BDGM20,GP20,WW20} follow the \emph{split fully homomorphic encryption} (split-FHE) approach~\cite{EC:BDGM20}, which we loosely recall here. The obfuscator computes
$$
\mathsf{FHE}(C) \xrightarrow[]{\mathsf{Eval}(x, \cdot)} \mathsf{FHE}(C(x)) \xrightarrow[]{\mathsf{Hint}(\sk, \cdot)} h
$$
where the decryption hint $h$ is specific to the ciphertext encoding $C(x)$. The obfuscated circuit consists of $(\mathsf{FHE}(C),h)$ and the evaluator can recompute the homomorphic evaluation and use the decryption hint $h$ to recover $C(x)$. It turns out that, as long as $|h| \ll |C(x)|$, this primitive alone is enough to build full-fledged obfuscation. One crucial aspect of this paradigm is that the split-FHE evaluation algorithm is deterministic, which allows the obfuscator and the evaluator to converge to the exact same ciphertext.

Translating this approach to the quantum setting seems to get stuck at a very fundamental level: Known FHE schemes for quantum circuits~\cite{FOCS:Mahadev18b} have an inherently randomized homomorphic evaluation algorithm. This means that the evaluator and the obfuscator would most likely end up with a different ciphertext even though they are computing the same function. This makes the hint $h$ (which is ciphertext-specific) completely useless to recover the output. A similar barrier also emerges in other generic transformations, such as~\cite{STOC:GKPVZ13}, which are at the foundations of many existing obfuscation schemes.

\paragraph{Reducing to Classical Obfuscation.} As a direct approach seems to be out of reach of current techniques, in this work we take a different route. Following the template of~\cite{FOCS:GGHRSW13}, our high-level idea is to outsource the quantum computation to some \emph{untrusted} component of the scheme and instead obfuscate only the circuit that \emph{verifies} that the computation was carried out correctly.
The important point here is that verifying the correctness of quantum computation is a much easier task than performing the computation itself. In fact, it was recently shown~\cite{FOCS:Mahadev18a} that the validity of any BQP computation can be verified by a completely classical algorithm, assuming the quantum hardness of the LWE problem. Furthermore, recent works~\cite{TCC:ACGH20,TCC:ChiChuYam20} have shown that the protocol can be collapsed to two rounds (in the random oracle model): On input a quantum circuit $Q$, the verifier produces some public parameters $\params$, which can be used by the prover (holding a quantum state $\ket{\psi}$) to compute a classical proof $\pi$. The verifier can then locally verify $\pi$ using some secret information $r$ which was sampled together with $\params$.

Using these classical verification of quantum computation (CVQC) protocols without any additional modification would allow us to implement the scheme as outlined above. However, we would not get any meaningful notion of privacy for the obfuscated circuit, since the prover needs to evaluate the circuit $Q$ in plain. Thus, we will need to turn these protocol \emph{blind}: The prover is able to prove that $Q(\ket{\psi}) = y$ obliviously, without knowing $Q$. This can be done in a canonical way, using fully homomorphic encryption for quantum circuits (QFHE) with classical keys~\cite{FOCS:Mahadev18b,C:Brakerski18}.




\paragraph{Challenges Towards Provable Security.} Although everything seems to fall in place, there is a subtle aspect that makes our attempt not sound: We implicitly assumed that the CVQC protocol is \emph{resettably secure}. If the prover is given access to a circuit implementing a (obfuscated) verifier, nothing prevents it from rewinding it in an attempt to extract the verifier's secret. Once this is leaked, the prover can fool the verifier into accepting false statements, ultimately learning some information about the obfuscated circuit. This is not only a theoretical concern, instead one can show concrete attacks against all known CVQC protocols (more discussion on this later). While this class of attacks seems to be hard to prevent in general, we observe that, if we restrict our attention to null (i.e.\ always rejecting) circuits, then this concern disappears. This is not a coincidence: Any two-round CVQC protocol that is one-time sound is automatically many-time sound for reject-only circuits, since it is easy to simulate the responses of the verifier (always reject). 

Another challenge that we need to resolve is that of provable security: The classical obfuscation only provides us with the weak guarantee of computational indistinguishability for functionally equivalent (classical) circuits. Even if we restrict our attention to null circuits, our scheme still needs to hardwire the verifier's secret in the obfuscated verifier circuit. That is, to obfuscate a null circuit $Q$, we publish
$$
\params\text{ and } \obf\left(\Pi_{\mathsf{nullQiO}}(\pi): \text{Return } \mathsf{CVQC}.\Verify(\params, \pi, r)\right)
$$
where $(\params, r) \sample \mathsf{CVQC}.\KeyGen(1^\lambda, Q)$ is sampled by the obfuscator. To show security, we cannot simply switch the obfuscated circuit to reject all inputs: Since the CVQC protocol is only \emph{computationally sound}, valid proofs $\pi$ for false statements always exist, they are just hard to find. In particular, this means that the circuit $\Pi_{\mathsf{nullQiO}}$ as defined above is not functionally equivalent to an always rejecting (classical) circuit.

\paragraph{Quantum Null iO from Trapdoor Dual-Mode CVQC.} To address the above issue, we introduce the notion of \emph{dual-mode} CVQC. As mentioned earlier, a dual-mode CVQC supports an alternative parameter generation algorithm $\SKeyGen$ where for any \emph{null} circuit $Q$ and $(\params,r) \gets \mathsf{CVQC}.\SKeyGen(1^\secp,Q)$, there do not exist any proofs $\pi$ that accept with respect to the verification key $r$. Furthermore, the parameters and verification key generated by $\SKeyGen$ should be computationally indistinguishable from those generated by $\KeyGen$. 

Unfortunately, the two-message CVQC protocol mentioned above does not satisfy this dual-mode property.\footnote{Note in particular that such a property actually implies publicly-verifiable CVQC, for which there were no known constructions prior to this work.} However, we observe that a weaker property, which we call the \emph{trapdoor} dual-mode property, both suffices for quantum null-iO, and can be shown to exist in the quantum random oracle model. In a trapdoor dual-mode CVQC, the standard key generation algorithm $\KeyGen$ does not support the dual-mode property. However, there exists a ``trapdoor'' parameter generation algorithm $\mathsf{CVQC}.\TKeyGen$ that returns the public parameters $\params_\td$ along with a secret key $r_\td$ and a trapdoor $\td$. The operation of the protocol in this trapdoor setting does actually satisfy the dual-mode property. In full detail, a trapdoor dual-mode CVQC satisfies the following properties:
\begin{itemize}
    \item (Setup Indistinguishability) For all circuits $Q$, the distributions $$\mathsf{CVQC}.\KeyGen(1^\lambda, Q) \approx_c (\params_\td, r_\td)$$ are computationally indistinguishable, where $(\params_\td, r_\td, \td) \sample \mathsf{CVQC}.\TKeyGen(1^\lambda, Q)$.
    \item (Verification Equivalence) The algorithms $$\mathsf{CVQC}.\Verify(Q, \cdot, r_\td) \equiv \mathsf{CVQC}.\TVerify(Q, \cdot, \td)$$ are functionally equivalent.
\end{itemize}

Moreover, in the trapdoor setting the scheme should satisfy the dual-mode property, using parameter generation algorithm $\SKeyGen$.

\begin{itemize}
    \item (Dual-Mode) For any null circuit $Q$, the distributions $$(\params_\td,\td) \approx_c (\params_\simgen,\td_\simgen)$$ are computationally indistinguishable, where $(\params_\td,\sk_\td,\td) \sample \mathsf{CVQC}.\TKeyGen(1^\secp,Q)$ and $(\params_\simgen,\td_\simgen)\allowbreak \sample \mathsf{CVQC}.\SKeyGen(1^\secp,Q)$. Moreover, the circuit $\mathsf{CVQC}.\TVerify(Q,\cdot,\td_\simgen)$ has no accepting inputs.
\end{itemize}


Deferring for the moment the discussion on how to actually construct a trapdoor dual-mode CVQC, we now argue that the above properties suffice for constructing quantum null-iO. Our obfuscation scheme will make use of quantum fully-homomorphic encryption to make the CVQC blind, as well as classical indistinguishability obfuscation to hide the secret key $r$ of the CVQC scheme and the secret key $\sk$ of the QFHE scheme. An obfuscation of circuit $Q$ consists of

$$\QFHE.\Enc(\params), \ \  \obf(\Pi_{\mathsf{nullQiO}}(\ct_\pi) : \text{ Return } \mathsf{CVQC}.\Verify(Q,\QFHE.\Dec(\sk,\ct_\pi),r)),$$ where $(\params,r) \sample \mathsf{CVQC}.\KeyGen(1^\secp,Q)$.
To show indistinguishability security, we use the properties of trapdoor dual-mode CVQC to gradually move from an obfuscation of null circuit $Q_0$ to null circuit $Q_1$: 
\begin{align*}
    &\QFHE.\Enc(\params), \ \  \obf(\Pi_{\mathsf{nullQiO}}(\ct_\pi) : \text{ Return } \mathsf{CVQC}.\Verify(Q_0,\QFHE.\Dec(\sk,\ct_\pi),r)) \\
    &\approx_c \QFHE.\Enc(\fbox{$\params_\td$}), \ \  \obf(\Pi_{\mathsf{nullQiO}}(\ct_\pi) : \text{ Return } \mathsf{CVQC}.\Verify(Q_0,\QFHE.\Dec(\sk,\ct_\pi),\fbox{$r_\td$}))\\
    &\approx_c \QFHE.\Enc(\params_\td), \ \  \obf(\Pi_{\mathsf{nullQiO}}(\ct_\pi) : \text{ Return } \mathsf{CVQC}.\Verify(Q_0,\QFHE.\Dec(\sk,\ct_\pi),\fbox{$\td$}))\\
    &\approx_c \QFHE.\Enc(\fbox{$\params_\simgen$}), \ \  \obf(\Pi_{\mathsf{nullQiO}}(\ct_\pi) : \text{ Return } \mathsf{CVQC}.\Verify(Q_0,\QFHE.\Dec(\sk,\ct_\pi),\fbox{$\td_\simgen$}))\\
    &\approx_c \QFHE.\Enc(\params_\simgen), \ \  \obf(\fbox{$\Pi_\bot$})
\end{align*}
where $\Pi_\bot$ is the always rejecting circuit. At this point we can appeal to semantic security of the QFHE to switch the parameter generation to use $Q_1$, and then undo the above sequence of four hybrids. This establishes that the obfuscation of $Q_0$ is computationally indistinguishable from the obfuscation of $Q_1$.


Finally, applying a known transformation, we obtain a witness encryption scheme for QMA as an immediate corollary.

\paragraph{A Trapdoor Dual-Mode CVQC Protocol in the QROM.} What is left to be shown is an instantiation of a trapdoor dual-mode CVQC protocol. While we do not know of a scheme secure in the standard model (in fact, we do not yet know how to construct \emph{any} two-message CVQC secure in the standard model, even without the trapdoor dual-mode property), we present a construction secure against the LWE assumption in the QROM. In fact, we show how to compile any two-message CVQC protocol (such as the one discussed above), into a trapdoor dual-mode CVQC. We make use of a random oracle $\cH : \{0,1\}^\ast \to \{0,1\}^{\lambda + 1}$ that may be queried in superposition. The modified protocol simply consists of running the CVQC prover and hashing the resulting proof
$$
\pi \sample \mathsf{CVQC}.\Prove(\params, \ket{\psi}) \text{ and } h = \cH(\pi)
$$
whereas the verification algorithm checks the consistency of $h \stackrel{?}{=} \cH(\pi)$, in addition to running the $\mathsf{CVQC}.\Verify$ algorithm on $\pi$.

In the trapdoor setting, we move the computation of the verification algorithm $\mathsf{CVQC}.\Verify(\params,\pi,r)$ into the specification of the random oracle $\cH$. That is, we replace the last bit of the random oracle output with an encryption of $\mathsf{CVQC}.\Verify(\params,\pi,r)$, under a secret key $\td$ that functions as the trapdoor. Then, the trapdoor verification algorithm no longer requires the CVQC secret parameters: It can instead use $\td$ to decrypt the last bit of the random oracle output $h$ in order to uncover the result of $\mathsf{CVQC}.\Verify(\params,\pi,r)$.  To implement this, we use a quantum-secure PRF $F$ with key $\td$. On input some proof $\pi$, we set the first $\lambda$ bits of the random oracle to be uniformly sampled and the last bit to
$$F(\td,\pi) \oplus \mathsf{CVQC}.\Verify(\params,\pi,r).$$
Note that the verifier ($\TVerify$) can equivalently check the validity of $\pi$ by simply recomputing $F(\td,\pi)$ and unmasking the response that was already computed in the random oracle.

Now, it remains to show how we obtain the dual-mode property in the trapdoor setting. This follows by letting $\SKeyGen$ simply be the same as $\TKeyGen$, except that $\mathsf{CVQC}.\Verify(\params,\pi,r)$ is replaced with 0 always in the implementation of the random oracle. That is, the last bit of the random oracle output is always $F(\td,\pi)$, and thus, verification using $\td$ will always output 0. To show that this is computationally indistinguishable from $\mathsf{CVQC}.\TKeyGen$ for any null circuit $Q$, we observe that any adversary that can distinguish these oracles can be used to break the soundness of $\mathsf{CVQC}$. This follows by specifying a reduction that measures one of the adversary's oracle queries to obtain an accepting proof with noticeable probability.



\subsubsection{Applications}

Next, we explore some applications of our newly constructed null-iO for quantum circuits. Since it is the weaker primitive, we are going to use witness encryption for QMA (with classical ciphertexts) as the starting point for all of our primitives. 

\paragraph{NIZK for QMA.} Our first result is a construction of NIZK arguments for QMA with public verifiability. To build up some intuition about the protocol, consider the simplified setting where we have a single fixed statement $x$. We can then define the common reference string to be
$$
(\mathsf{vk}, \WE.\Enc(x, \sigma)) \text{ such that } \mathsf{Verify}(\mathsf{vk}, \sigma, x) = 1
$$
where $\mathsf{vk}$ is a verification key of a signature scheme. Anyone with a valid witness $\ket{\psi}$ for $x$ can recover the signature by decrypting the ciphertext. Then anyone can verify the validity of $x$ by simply verifying the signature $\sigma$ against $\vk$. Note that this computation is entirely classical and succinct: It's runtime does not depend on the computation needed to compute $\ket{\psi} \in \relation(x)$.
To extend this approach to an exponential number of statements, we exploit the fact that our witness encryption scheme has a \emph{completely classical} encryption procedure. Our idea is to place in the common reference string the obfuscation of the (classical) circuit 
$$
\Pi_\NIZK[\sk,k](x): \text{Return } \WE.\Enc(x, \sigma; \prf(k,x))
$$ 
where $\prf$ is a puncturable PRF~\cite{STOC:SahWat14}. The prover algorithm can then evaluate the obfuscated circuit on $x$ to obtain $\WE.\Enc(x, \sigma)$ and proceed as before. Such an approach can be shown to be sound via a standard puncturing argument.

\paragraph{ZAPR for QMA.}
The next question that we ask is whether we can reduce the trust on the setup and obtain some meaningful guarantees also in the presence of a maliciously generated common reference string. Known generic transformations~\cite{FOCS:DwoNao00} do not apply to our case, since the NIZK that we obtain is an argument, i.e.\ it has computational soundness. Our saving grace is again the fact that our NIZK has completely classical setup and verification procedures. This allows us to leverage powerful tools from the literature of (classical) zero-knowledge. We adopt a dual-track approach (reminiscent of the Naor-Yung~\cite{STOC:NaoYun90} paradigm) where we define the setup to sample two copies $(\crs_0, \crs_1)$, the image of a one-way function $y$, and a \emph{classical} non-interactive witness indistinguishable (NIWI) proof that either $\crs_0$ or $\crs_1$ is correctly generated. 

At this point, it is still unclear whether we have any privacy guarantee, since one of the two strings can be maliciously generated and can therefore leak some information about the witness, if naively used by the prover. Thus, instead of having the prover directly compute the proofs  $(\pi_0, \pi_1)$, we let it compute a classical NIWI for the statement
$$
\left\{
\exists~(\pi_0, \pi_1, z) \text{ such that: }
\begin{array}{l}
\mathsf{Verify}(\crs_0, \pi_0, x) = 1 \text{ OR }\\
\mathsf{Verify}(\crs_1, \pi_1, x) = 1 \text{ OR }\\
\mathsf{OWF}(z) = y.
\end{array}
\right\}.
$$
Since the verification algorithm of our NIZK schem is classical, then so is the above statement.
By the witness indistinguishbility of the NIWI, the verifier cannot distinguish whether the prover inverted the one-way function or possesses a valid proof. Proving soundness requires more work, since our NIZK scheme is only \emph{computationally sound}. To get around this, we further augment the scheme with a statistically hiding sometimes-binding (SBSH) commitment~\cite{EC:KalKhuSah18,EC:GJJM20,EC:BFJKS20}: This tool allows the prover to commit to its witness, which is statistically hidden, except with some (exponentially) small probability where the commitment is efficiently extractable. We can then set the parameters of our primitives to be sufficiently large (i.e.\ use complexity leveraging) to ensure that whenever the extraction even happens, it still leads to a contradiction to the soundness of the NIZK or to the one-wayness of $\mathsf{OWF}$.

\paragraph{ABE for BQP.} We next show how our witness encryption scheme for QMA yields the first ABE scheme for quantum functionalities. For starters, consider again the simplified setting where the key authority issues a single key for a fixed attribute $x$. The witness encryption suggests a natural ABE encryption procedure for a policy encoded by a quantum circuit $Q$: We compute $\WE.\Enc((Q,x), m)$, where the statement $(Q,x)$ returns $1$ if and only if $Q(x)=1$. Note that this is technically a BQP statement (the witness is publicly computable), but a witness encryption for QMA is also a witness encryption for BQP. The challenge is now to define an encryption algorithm that does not need to take the attribute $x$ as an input. Instead of publishing the ciphertext directly, the encrypter obfuscate the (classical) circuit
$$
\Pi_\ABE[\vk,k](x, \sigma): \text{If } \mathsf{Verify}(\mathsf{vk}, \sigma, x) = 1 \text{ return } \WE.\Enc((Q,x), m; \prf(k, x))
$$
where $\prf$ is a puncturable PRF and $\vk$ is the verification key for a signature scheme. A secret key for an attribute $x$ simply consists of a signature $\sigma_x$ on $x$, computed with a signing key held by the key authority. This way, the holder of a key for an attribute $x$ can only extract witness encryption ciphertexts associated with $x$. Some additional work is needed in order to obtain a provably secure scheme, but the main ideas are already present in this outline.

The scheme as described so far does not hide the circuit $Q$, i.e.\ it only satisfies the notion of payload hiding. We show a generic compiler that transforms any quantum ABE with payload-hiding security into one with predicate encryption security, i.e.\ where the circuit $Q$ is hidden to the holders of keys for rejecting attributes. This result is obtained by introducing the notion of \emph{quantum lockable obfuscation} and presenting a construction under LWE.

\paragraph{Constrained PRF for BQP.} Given the above ABE scheme for BQP, one can easily turn any puncturable PRF into a constrained PRF. For convenience, here we consider a key-policy ABE, which can be obtained from the scheme as described above via universal (quantum) circuits. The public parameters of the PRF are augmented with an obfuscated circuit
$$
\Pi_\prf[k, \tilde{k}](x): \text{Return } \ABE.\Enc(x, \PRF(k, x); \PRF(\tilde{k}, x))
$$
where $\tilde{k}$ is an independently sampled key. Note that anyone can query such circuit on any attribute $x$, however only the holder of a key for a policy $Q$ such that $Q(x)=1$ (with overwhelming probability) can recover the PRF output $\PRF(k, x)$ by decrypting the resulting ciphertext. Furthermore, observe that the functionality specified above is entirely classical, and therefore classical obfuscation suffices.




\paragraph{Secret Sharing for Monotone QMA.} It is well-known that witness encryption for NP implies the existence of a secret sharing scheme for monotone NP~\cite{AC:KomNaoYog14}. The high-level idea is to assign to each party $P_i$ the opening of a perfectly binding commitment $c_i = \mathsf{Com}(i; r_i)$ encoding the index corresponding to the party. Then one can publish a witness encryption for the statement
$$
\left\{\exists~ (I\subseteq P, r_1, \dots, r_{|I|}) \text{ such that: } I \in \lang \text{ AND } \forall i\in I: c_i = \mathsf{Com}(i;r_i)\right\}
$$
where $I$, parsed as a binary string, forms a statement in a NP-complete language $\lang$ with witness $w$. It is not hard to show that decrypting the witness encryption (i.e.\ reconstructing the secret) can only be done by an authorized set of parties holding the witness $w$. We show that this construction naturally generalizes to the QMA setting, when given a witness encryption scheme for QMA.

\subsection{Discussion and Open Problems}

We discuss two clear open problems that are suggested by this work. We identify barriers towards making progress on each problem with our current approach.

\paragraph{Obfuscation Beyond Null Circuits.}
In this work, we only consider obfuscating the CVQC verification circuit in the setting where each instance the prover can query will be rejecting (with high probability). One could also consider obfuscating the verification circuit in the setting where the prover can query on an accepting instance, which would help in constructing fully-fledged iO for \emph{all} quantum circuits. We expect obfuscation for general quantum circuits to have a variety of applications and we consider it a fascinating problem in its own right.

Unfortunately, it turns out that this approach is in general insecure and concrete attacks exist against all known constructions of CVQC. We provide a high-level description of these attacks in Section~\ref{sec:cryptanalysis}. The main source of trouble appears to be the lack of \emph{resettable security} of CVQC protocols.  That is, an attacker is able to extract the verifier's secret by observing its responses on accepting instances.

\paragraph{Quantum Null-iO from Standard Assumptions.}
Another natural question is whether one can obtain quantum null-iO from standard cryptographic assumptions. Since we only give a construction of (trapdoor) dual-mode CVQC in the QROM, the resulting quantum null-iO does not achieve provable security. We stress that even without the dual-mode property, two-message CVQC protocols are only known in the QROM. One approach towards evading this barrier could be to instantiate the base CVQC protocol with the two message protocol of \cite{morimae2021classically} (with quantum first message), which is statistically sound. This would result in a valid quantum null-iO since the first message can be computed by the obfuscator and sent along with the obfuscated circuit. However, even in this case, the verification circuit in \cite{morimae2021classically} will accept exponentially many proofs even for no instances, as the underlying delegation of quantum computation protocol has a probabilistic verifier (that chooses which Hamiltonian terms to measure on each copy of the history state). Even though soundness can be driven to negligible by parallel repetition, this also rapidly increases the proof size. Thus, attempting to hybrid over each proof will fail, since the number of hybrids will be much larger than inverse of the soundness error. Given this barrier, we leave constructing any of the primitives discussed in this work in the standard model and against standard cryptographic assumptions, as an intriguing open problem.

In Appendix~\ref{sec:vbb}, we present an alternative construction of quantum null-iO assuming classical virtual-blackbox (VBB) obfuscation. While it is known that VBB obfuscation is in general impossible~\cite{C:BGIRSVY01}, classical VBB obfuscation has been used as a heuristic method to analyze the security of certain schemes. Recent examples include fully-homomorphic encryption for RAM programs~\cite{C:HHWW19} and one-shot signatures~\cite{STOC:AGKZ20}. As another example, Aaronson and Christiano~\cite{STOC:AarChr12} made use of ideal classical obfuscation to establish the feasibility of public-key quantum money. This influential result inspired a fruitful line of research, including a result by Zhandry~\cite{EC:Zhandry19b} that showed how to instantiate their original approach from indistinguishability obfuscation. 


\ifsubmission
\section{Preliminaries}
We denote by $\lambda$ the security parameter. A function $f : \mathbb{N} \rightarrow [0, 1]$ is negligible if for every constant $c \in \mathbb{N}$ there exists $N \in \mathbb{N}$ such that for all $n > N$, $f(n) < n^{-c}$. We recall some standard notation for classical Turing machines and Boolean circuits:
\begin{itemize}
\item We say that a Turing machine (or algorithm) is PPT if it is probabilistic and runs in polynomial time in $\lambda$.
\item We sometimes think about PPT Turing machines as polynomial-size uniform families of circuits. A polynomial-size circuit family $C$ is a sequence of circuits
$C = \{C_\lambda\}_{\lambda \in \mathbb{N}}$, such that each circuit $C_\lambda$ is of polynomial size $\lambda^{O(1)}$ and has $\lambda^{O(1)}$ input and output bits. We say that the family is uniform if there exists a polynomial-time deterministic Turing machine $M$ that on input $1^\lambda$ outputs $C_\lambda$.
\item For a PPT Turing machine (algorithm) $M$, we denote by $M(x; r)$ the output of $M$ on input $x$ and random coins $r$. For such an algorithm, and any input $x$, we write $m \in M(x)$ to denote that $m$ is in the support of $M(x;\cdot)$. Finally we write $y \sample M(x)$ to denote the computation of $M$ on input $x$ with some uniformly sampled random coins.
\end{itemize}

\subsection{Quantum Computation}

We recall some notation for quantum computation and we define the notions of computational and statistical indistinguishability for quantum adversaries. Various parts of what follows are taken almost verbatim from~\cite{STOC:BitShm20}.

We say that a Turing machine (or algorithm) is QPT if it is quantum and runs in polynomial time. We sometimes think about QPT Turing machines as polynomial-size uniform families of quantum circuits (as these are equivalent models). We call a polynomial-size quantum circuit family $C = \{C_\lambda\}_{\lambda \in \mathbb{N}}$ uniform if there exists a polynomial-time deterministic Turing machine $M$ that on input $1^\lambda$ outputs $C_\lambda$.

Throughout this work, we model efficient adversaries as quantum circuits with non-uniform quantum advices. This is denoted by $\adv^* = \{\adv^*_\lambda, \rho_\lambda\}_{\lambda \in \mathbb{N}}$, where $\{\adv^*_\lambda\}_{\lambda \in \mathbb{N}}$ is a polynomial-size non-uniform sequence of quantum circuits, and $\{\rho_\lambda\}_{\lambda \in \mathbb{N}}$ is some polynomial-size sequence of mixed quantum states. We now define the formal notion of computational indistinguishability in the quantum setting, where random variables $X,Y$ are represented as mixed quantum states.

\begin{definition}[Computational Indistinguishability]
Two ensembles of quantum random variables $\mathcal{X} = \{X_\lambda\}_{\lambda \in \mathbb{N}}$ and $\mathcal{Y} = \{Y_\lambda\}_{\lambda \in \mathbb{N}}$ are said to be
 computationally indistinguishable (denoted by $\mathcal{X} \approx_c \mathcal{Y}$) if there exists a negligible function $\nu$ such that for all $\lambda \in \mathbb{N}$ and all non-uniform
 QPT distinguishers with quantum advice 
 $\adv = \{\adv_\lambda, \rho_\lambda\}_{\lambda \in \mathbb{N}}$, it holds that
 \[
 \left|\Pr[\adv_\secp(X_\secp; \rho_\secp) = 1] - \Pr[\adv(Y_\secp; \rho_\secp) = 1]\right| \leq \nu(\lambda).
 \]
\end{definition}
The trace distance between two quantum distributions $(X_\lambda, Y_\lambda)$, denoted by $\mathsf{TD}(X_\lambda, Y_\lambda)$, is a generalization of statistical distance to the quantum setting and represents the maximal distinguishing advantage between two quantum distributions by an unbounded quantum algorithm. We define below the notion of statistical indistinguishability.
\begin{definition}[Statistical Indistinguishability]
Two ensembles of quantum random variables $\mathcal{X} = \{X_\lambda\}_{\lambda \in \mathbb{N}}$ and $\mathcal{Y} = \{Y_\lambda\}_{\lambda \in \mathbb{N}}$ are said to be
 statistically indistinguishable (denoted by $\mathcal{X} \approx_s \mathcal{Y}$) if there exists a negligible function $\nu$ such that for all $\lambda \in \mathbb{N}$, it holds that
 \[
    \mathsf{TD}(X_\lambda, Y_\lambda) \leq \nu(\lambda).
 \]
\end{definition}

Throughout this work, we will often consider quantum circuits $Q$ that output a single classical bit. Any such circuit can be written as a unitary followed by a computational basis measurement of the first qubit. When we compute $Q$ on some classical (resp. quantum) input $x$ (resp. $\ket{\psi}$), we write $Q(x)$ (resp. $Q(\ket{\psi})$) to denote the output that results from padding the input with sufficiently many ancillary $\ket{0}$ states (as determined by the description of $Q$), computing a unitary, and then measuring the first qubit. We will sometimes consider the following restricted families of ``psuedo-deterministic'' quantum circuits.

\begin{definition}[Pseudo-Deterministic Quantum Circuit]
\label{def: deterministic}
A family of psuedo-deterministic quantum circuits is defined by a family of circuits $\{Q_\secp\}_{\secp \in \mathbb{N}}$. The circuit defined by $Q_\secp$ takes as input a bit string $x \in \{0,1\}^{n(\secp)}$ (along with ancillary 0 states) and outputs a single classical bit $b \gets U(x)$. The circuit is pseudo-deterministic if there exists a negligible function $\nu$ such that for every sequence of classical inputs $\{x_\secp\}_{\secp \in \mathbb{N}}$, there exists a sequence of outputs $\{b_\secp\}_{\secp \in \mathbb{N}}$ such that \[\Pr[Q_\secp(x_\secp) = b_\secp] = 1-\nu(\secp).\] 
\end{definition}



\subsection{Learning with Errors}

We recall the definition of the learning with errors (LWE) problem~\cite{STOC:Regev05}.
\begin{definition}[Learning with Errors]
The LWE problem is parametrized by a modulus $q = q(\lambda)$, polynomials $n=n(\lambda)$ and $m=m(\lambda)$, and an error distribution $\chi$. The LWE problem is hard if it holds that
\[
(\mathbf{A},\mathbf{A} \cdot \mathbf{s} + \mathbf{e}) \approx_c (\mathbf{A},\mathbf{u})
\]
where $\mathbf{A}\sample\ZZ_q^{m \times n}$, $\mathbf{s}\sample\ZZ_q^n$, $\mathbf{u}\sample\ZZ_q^m$, and $\mathbf{e}\sample\chi^m$.
\end{definition}
As shown in \cite{STOC:Regev05,STOC:PeiRegSte17}, for any sufficiently large modulus $q$ the LWE problem where $\chi$ is a discrete Gaussian distribution with parameter $\sigma = \xi q \ge 2 \sqrt{n}$ (i.e.\ the distribution over $\ZZ$ where the probability of $x$ is proportional to $e^{-\pi (|x|/\sigma)^2}$), is at least as hard as approximating the shortest independent vector problem (SIVP) to within a factor of $\gamma = \tilde{O}({n}/\xi)$ in \emph{worst case} dimension $n$ lattices.

\subsection{Quantum Fully-Homomorphic Encryption}

We recall the notion of quantum fully homomorphic encryption (QFHE)~\cite{C:BroJef15}. In this work we are interested in QFHE schemes with classical keys, classical encryption of classical messages, and classical decryption and therefore we only define QFHE for this restricted case.

\begin{definition}[Quantum Homomorphic Encryption]\label{def:qfhe}
A quantum homomorphic encryption scheme $(\QFHE.\Gen,\allowbreak\QFHE.\Enc,\allowbreak \QFHE.\Eval,\allowbreak \QFHE.\Dec)$ consists of the following efficient algorithms.
\begin{itemize}
    \item $\QFHE.\Gen(1^\lambda)$: On input the security parameter, the key generation algorithm returns secret/public key pair $(\sk, \pk)$.
    \item $\QFHE.\Enc(\pk, m)$: On input the public key $\pk$ and a message $m$, the encryption algorithm returns a ciphertext $c$.
    \item $\QFHE.\Eval(\pk, C, c)$: On input the public key $\pk$, a quantum circuit $C$, and a ciphertext $c$, the evaluation algorithm returns an evaluated  ciphertext $\Tilde{c}$.
    \item $\QFHE.\Dec(\sk, c)$: On input the secret key $\sk$ and a ciphertext $c$, the decryption algorithm returns a message $m$.
\end{itemize}
\end{definition}
Analogously to the classical case~\cite{STOC:Gentry09}, we say that the scheme is fully homomorphic if the evaluation algorithm supports all polynomial-size quantum circuits. We say that the scheme is leveled if the maximum depth of the homomorphically evaluated circuits is bounded by the key generation algorithm (formally, the key generation algorithm takes as input an additional $1^d$ depth parameter). In this work, we are only interested in levelled schemes, so we drop the adjective wherever it is clear from the context. Next we define the notion of (single-hop) evaluation correctness for QFHE.

\begin{definition}[Evaluation Correctness]
A quantum homomorphic encryption scheme $(\QFHE.\Gen,\allowbreak\QFHE.\Enc,\allowbreak \QFHE.\Eval,\allowbreak \QFHE.\Dec)$ is correct if for all $\lambda \in\mathbb{N}$, all $(\sk, \pk)\in\QFHE.\Gen(1^\lambda)$, all messages $m$, and all polynomial-size quantum circuits $C$, it holds that
\[
\QFHE.\Dec(\sk, \QFHE.\Eval(\pk, C, \QFHE.\Enc(\pk,m))) \approx_s C(m).
\]
\end{definition}
The notion of semantic security is defined analogously to the classical case, and we refer the reader to~\cite{C:BroJef15} for a formal definition. 
The works of Mahadev~\cite{FOCS:Mahadev18b} and Brakerski~\cite{C:Brakerski18} show that QFHE with classical keys can be constructed from the quantum hardness of the LWE problem. We recall their results below.
\begin{lemma}[\cite{FOCS:Mahadev18b,C:Brakerski18}]
Assuming the quantum hardness of the LWE problem, there exists a (leveled) QFHE scheme $(\QFHE.\Gen,\allowbreak\QFHE.\Enc,\allowbreak \QFHE.\Eval,\allowbreak \QFHE.\Dec)$ with classical keys and classical decryption.
\end{lemma}

\else

\ifsubmission\section{Additional Preliminaries}\else\fi

\subsection{Program Obfuscation}

We recall two notions of obfuscation for classical circuits: indistinguishability obfuscation and virtual black-box obfuscation. An obfuscator $\obf$ is a PPT algorithm that takes as input a circuit $C$ and outputs an obfuscated circuit $\widetilde{C}$. For both notions of obfuscation, $\obf$ must satisfy the following correctness property.
\begin{definition}[Correctness]
An obfuscator $\obf$ is correct if for all circuits $C$, it holds that
\[
\Pr\left[\forall x \in \{0,1\}^n: C(x)= \obf(1^\lambda, C)(x)\right] = 1.
\]
\end{definition}

Now we define the two notions of security. First, we define \emph{indistinguishability obfuscation}.
\begin{definition}[Indistinguishability Obfuscation]
For all pairs of circuits $(C_0, C_1)$ such that $|C_0| = |C_1|$, and such that for all inputs $x$, it holds that $C_0(x) = C_1(x)$, it holds that
\[
\obf(1^\lambda, C_0) \approx_c \obf(1^\lambda, C_1).
\]
\end{definition}
This definition extends to the post-quantum setting by allowing the adversary to be QPT.

\subsection{Puncturable Pseudorandom Functions}

We recall the definition of a puncturable pseudorandom function (PRF) from~\cite{STOC:SahWat14}.
\begin{definition}[Puncturable Pseudorandom Function]
A puncturable PRF $(\PRF.\Gen, \PRF.\Puncture,\allowbreak \PRF.\Eval)$ consists of the following efficient algorithms.
\begin{itemize}
    \item $\PRF.\Gen(1^\lambda)$: On input the security parameter, the key generation algorithm returns a key $k$.
    \item $\PRF.\Puncture(k,z)$: On input a key $k$ and a point $z$, the puncturing algorithm returns the punctured key $k_z$.
    \item $\PRF.\Eval(k, x)$: On input a key $k$ and a string $x\in\{0,1\}^\lambda$, the evaluation algorithm returns a string $y\in\{0,1\}^{\lambda}$.
\end{itemize}
\end{definition}
Correctness requires that the evaluation over the punctured key agrees with the evaluation over the non-punctured key, except for the punctured point.
\begin{definition}[Correctness]
A puncturable PRF $(\PRF.\Gen, \PRF.\Puncture,\allowbreak \PRF.\Eval)$ is correct if for all $\lambda\in\mathbb{N}$, all $z \in\{0,1\}^\lambda$, and all strings $x\neq z$ it holds that
\[
\Pr\left[\PRF.\Eval(k, x) = \PRF.\Eval(k_z, x)\right] =1
\]
where $k \sample \PRF.\Gen(1^\lambda)$ and $k_z \sample \PRF.\Puncture(k, z)$.
\end{definition}
Pseudorandomness requires that the evaluation of the PRF at any point $z$ is computationally indistinguishable from random, even given the punctured key $k_z$.
\begin{definition}[Pseudorandomness]
A puncturable PRF $(\PRF.\Gen, \PRF.\Puncture,\allowbreak \PRF.\Eval)$ is pseudorandom if for all $\lambda\in\mathbb{N}$ and all $z \in\{0,1\}^\lambda$ it holds that
\[
(\PRF.\Eval(k, z), k_z) \approx_c (u, k_z)
\]
where $k \sample \PRF.\Gen(1^\lambda)$, $k_z \sample \PRF.\Puncture(k, z)$, and $u \sample \{0,1\}^\lambda$.
\end{definition}
\fi

\section{Obfuscation of Null Quantum Circuits}\label{sec:qiO}

In this section, we present a scheme to obfuscate null quantum circuits. Then we show that such a scheme implies the existence of a witness encryption scheme for QMA.

\subsection{Classical Verification of Quantum Computation}

Mahadev \cite{FOCS:Mahadev18a} gave the first protocol for classical verification of quantum computation. The protocol was recently improved to only require two messages in the quantum random oracle model~\cite{TCC:ChiChuYam20,TCC:ACGH20}. 


\begin{definition}[Classical Verification of Quantum Computation (CVQC)]\label{def: classical argument} A two-message classical verification of quantum computation protocol consists of algorithms $(\KeyGen, \Prove, \Verify)$. The following is defined for quantum circuits $Q$ with one classical bit of output, and involves a random oracle $\cH$.
\begin{itemize}
    \item $\KeyGen(1^\secp,Q)$: On input the security parameter and a quantum circuit $Q$, the PPT algorithm $\KeyGen$ outputs public parameters $\params$ and a verification key $r$.
    \item $\Prove^\cH(\params,\ket{\psi})$: On input the public parameters $\params$ and a quantum state $\ket{\psi}$, the QPT $\Prove$ algorithm, with query access to random oracle $\cH$, outputs a proof $\pi$.
    \item $\Verify^\cH(Q,\pi,r)$: On input a quantum circuit $Q$, a proof $\pi$, and a verification key $r$, the PPT algorithm $\Verify$, with query access to a random oracle $\cH$, outputs a bit $b$ indicating acceptance or rejection.
\end{itemize}
\end{definition}
We require the following notion of correctness.
\begin{definition}[Correctness]
A CVQC protocol $(\KeyGen, \Prove, \Verify)$ is correct if for any negligible function $\nu$, there exists a polynomial $k$ and negligible function $\mu$ such that for all polynomial-size families of quantum circuits $\{Q_\secp\}_{\secp \in \bbN}$ and inputs $\{\ket{\psi_\secp}\}_{\secp \in \bbN}$ such that there exists $b \in \{0,1\}$ such that $\Pr[Q_\secp(\ket{\psi_\secp}) = b] \geq 1-\nu(\secp)$, it holds that
    $$
    \begin{array}{r}
    \Pr\left[\Verify^\cH(Q_\secp,\pi,r) = b: (\params,r) \sample \KeyGen(1^\secp,Q_\secp), \pi \sample \Prove^\cH(\params,\ket{\psi_\secp}^{\otimes k(\secp)})\right] = 1 - \mu(\secp).
    \end{array}$$
\end{definition}

\begin{remark}
We state correctness above with respect to (quantum circuit, input) pairs that either accept or reject with overwhelming probability. By standard QMA amplification, a protocol that satisfies this correctness guarantee can also be used to verify (quantum circuit, input) pairs that either accept with probability $\alpha$ or reject with probability $\beta$, where $\alpha$ and $\beta$ are separated by an inverse polynomial.
\end{remark}

\paragraph{Soundness.} We define the notion of soundness below.
\begin{definition}[Soundness]
A CVQC protocol $(\KeyGen, \Prove, \Verify)$ is sound if for any negligible function $\nu$, any polynomial-size family of quantum circuits $\{Q_\secp\}_{\secp \in \bbN}$ such that for all inputs $\{\ket{\psi_\secp}\}_{\secp \in \bbN}$, $\Pr[Q_\secp(\ket{\psi_\secp}) = 1] \leq \nu(\secp)$, and all QPT adversaries $\cA$, there exists a negligible function $\mu$ such that 
    $$\Pr\left[\begin{array}{c} \Verify^\cH(Q_\secp,\pi,r) = 1 \end{array} : \begin{array}{r}(\params,r) \sample \KeyGen(1^\secp,Q_\secp), \\ \pi \sample \cA^{\ket{\cH}}(\params) \end{array}\right] = \mu(\secp).$$
    The notation $\cA^{\ket{\cH}}$ indicates that $\cA$ has \emph{quantum} query access to the random oracle $\cH$.
\end{definition}

\paragraph{Trapdoor Dual-Mode CVQC.} In the following we define CVQC with a strong notion of soundness, which involves a \emph{dual-mode} verifier algorithm. 

\begin{definition}[Trapdoor Dual-Mode CVQC]
A trapdoor dual-mode CVQC protocol $(\KeyGen, \Prove, \Verify)$ includes an additional triple of algorithms $(\TKeyGen, \TVerify, \SKeyGen)$ with the following properties.
\begin{itemize}
    \item (Setup Indistinguishability) For any polynomial-size family of quantum circuits $\{Q_\secp\}_{\secp \in \bbN}$ it holds that 
    $$
    \KeyGen(1^\lambda, Q_\secp) \approx_c (\params_{\td},r_{\td})
    $$
    where $(\params_{\td},r_{\td}, \td) \sample \TKeyGen(1^\lambda, Q_\secp)$.
    \item (Verification Equivalence) For any polynomial-size family of quantum circuits $\{Q_\secp\}_{\secp \in \bbN}$, and any proof $\pi$, it holds that
    $$
    \Verify(Q_\secp,\pi,r_{\td}) = \TVerify(Q_\secp,\pi,\td)
    $$
    where $(\params_{\td},r_{\td}, \td) \sample \TKeyGen(1^\lambda, Q_\secp)$.
    
    \item (Dual-Mode)
    For any negligible function $\nu$, any polynomial-size family of quantum circuits $\{Q_{\secp}\}_{\secp \in \bbN}$ such that for all inputs $\{\ket{\psi_\secp}\}_{\secp \in \bbN}$, $\Pr[Q_{\secp}(\ket{\psi_\secp}) = 1] \leq \nu(\secp)$, it holds that 
    $$(\params_\td,\td) \approx_c (\params_\simgen,\td_\simgen),$$ where $(\params_{\td},r_{\td}, \td) \sample \TKeyGen(1^\lambda, Q_\secp)$, $(\params_\simgen,\td_\simgen) \gets \SKeyGen(1^\secp,Q_\secp)$, and $\TVerify(Q_\secp,\cdot,\td_\simgen)$ has no accepting input.
 \end{itemize}
\end{definition}
Let $(\KeyGen,\Prove,\Verify)$ be a two-message CVQC scheme (Definition~\ref{def: classical argument}). Our construction of trapdoor dual-mode CVQC is identical to such a two-message CVQC except for the modified prover algorithm
$$
\Prove^{\ast \cH}(\params,\ket{\psi}): \text{Return } (\pi \sample \Prove^{\cH}(\params,\ket{\psi}), \cH(\pi)),
$$
where $\cH$ is a random oracle with range $\{0,1\}^{\secp + 1}$, and the modified verification algorithm
$$
\Verify^{\ast \cH}(Q, (\pi, h), r): \text{If } \cH(\pi) \stackrel{?}{=} h \text{ return }\Verify^{\cH}(Q, \pi, r), \text{ else return } 0.
$$
In order to prove our construction, we are going to assume the existence of a quantum-secure PRF $F : \{0,1\}^\secp \times \{0,1\}^* \to \{0,1\}$, that is, a function that remains computationally indistinguishable from a truly random function even when the adversary can issue superposition queries. Quantum-secure PRFs are known from quantum-secure one-way functions~\cite{FOCS:Zhandry12}.

\begin{theorem}
Let $(\KeyGen,\Prove,\Verify)$ be a two-message CVQC scheme and let $F$ be a quantum-secure PRF. Then $(\KeyGen,\Prove^\ast,\Verify^\ast)$ is a trapdoor dual-mode CVQC in the QROM.
\end{theorem}

\begin{proof}
First, we define the algorithms $(\TKeyGen,\TVerify,\SKeyGen)$.
\begin{itemize}
\item $\TKeyGen(1^\secp,Q):$ Sample $(\params,r) \sample \KeyGen(1^\secp,Q)$ and a PRF key $k$. Set $\td = k$ and output $(\params,r,\td)$. Also, the random oracle $\cH$ will be implemented using a random oracle $\cG : \{0,1\}^* \to \{0,1\}^\secp$ and evaluated as $\cH(x) = (\cG(x), F(\td,x) \oplus \Verify(Q,x,r))$.
\item $\TVerify^\cH(Q,(\pi,h),\td):$ If $\cH(\pi) \stackrel{?}{=} h$ return $h_{\secp + 1} \oplus F(\td,\pi)$.
\item $\SKeyGen(1^\secp,Q):$ Same as $\TKeyGen$ except that the random oracle $\cH$ will be implemented as $\cH(x) = (\cG(x), F(\td,x))$.

\end{itemize}

We now argue that the three properties are satisfied.

To show Setup Indistinguishability, we first note that $(\params,r)$ are sampled identically by $\KeyGen$ and $\TKeyGen$. Thus, it suffices to argue that the adversary cannot notice the change in the implementation of the random oracle $\cH$, which we show via a sequence of hybrids.  Let $\cG: \{0,1\}^* \to \{0,1\}^\secp$ and $\cI : \{0,1\}^* \to \{0,1\}$ be random oracles.

\begin{itemize}
    \item $\Hyb_0:$ This is the original hybrid, where $\cH$ is implemented as a random oracle $\{0,1\}^* \to \{0,1\}^{\secp + 1}$.
    \item $\Hyb_1:$ Implement $\cH(x)$ as $(\cG(x),\cI(x))$, which is perfectly indistinguishable from $\Hyb_0$.
    \item $\Hyb_2:$ Implement $\cH(x)$ as $(\cG(x),\cI(x) \oplus \Verify(Q,x,r))$, which is perfectly indistinguishable from $\Hyb_1$.
    \item $\Hyb_3:$ Implement $\cH(x)$ as $(\cG(x), F(\td,x) \oplus \Verify(Q,x,r))$, which is computationally indistinguishable from $\Hyb_2$ due to the security of $F$.
\end{itemize}

Next, Verification Equivalence follows by definition. Finally, we show the Dual-Mode property. First, we note that the fact that $\TVerify(Q,\cdot,\td_\simgen)$ has no accepting input follows by definition. It remains to argue that $(\params_\td,\td) \approx_c (\params_\simgen,\td_\simgen)$, which we show via a sequence of hybrids. Let $q$ be an upper bound on the number of random oracle queries made by the adversary. 

\begin{itemize}
    \item $\Hyb_0:$ This is $(\params_\td,\td)$ sampled by $\TKeyGen(1^\secp,Q)$ and $\cH$ implemented as $\cH(x) = (\cG(x), F(\td,x) \oplus \Verify(Q,x,r))$.
    \item $\Hyb_{1,i}$ for $i \in [q]:$ The adversary's first $i$ queries to $\cH$ are answered as $\cH(x) = (\cG(x),F(\td,x))$. That is, in the first $i$ queries, $\Verify(Q,x,r)$ is always set to 0. The indistinguishability $\cH_{1,i} \approx_c \cH_{1,i-1}$ follows from the soundness of the CVQC protocol. Indeed, an adversary can only distinguish if its $i$'th query has some inverse polynomial amplitude on $x$ such that $\Verify^\cH(Q,x,r) = 1$. Otherwise, the hybrids would be statistically close. However, in this case, a reduction can produce an accepting proof for CVQC with inverse polynomial probability by answering each of the first $i-1$ queries as in $\cH_{1,i-1}$ and then measuring the $i$'th query. This violates the soundness of CVQC, since $Q$ rejects on all inputs with overwhelming probability.
\end{itemize}

Finally, note that $\Hyb_{1,q}$ is $\SKeyGen(1^\secp,Q)$, which completes the proof.

\end{proof}

\subsection{Definition}

Below we define obfuscation for a class of null quantum circuits. We again consider the set of quantum circuits that have one bit of classical output.

\begin{definition}[Null-iO for Quantum Circuits]\label{def: null-iO}
A null obfuscator $(\obf,\Eval)$ for quantum circuits is defined as the following QPT algorithms.
\begin{itemize}
    \item $\obf(1^\lambda, Q)$: The obfuscation algorithm takes as input a security parameter $1^\secp$ and a quantum circuit $Q$ with one classical bit of output, and outputs an obfuscation $\widetilde{Q}$.
    \item $\Eval(\widetilde{Q}, \ket{\psi})$: The evaluation algorithm takes as input a unitary $\widetilde{Q}$ and an input $\ket{\psi}$ and outputs a bit $b$.
\end{itemize}
\end{definition}
We define correctness below.
\begin{definition}[Correctness]
A null obfuscator $(\obf,\Eval)$ is correct if for any negligible function $\nu$, there exists a polynomial $k$ and negligible function $\mu$ such that for all polynomial-size families of quantum circuits $\{Q_\secp\}_{\secp \in \bbN}$ and inputs $\{\ket{\psi_\secp}\}_{\secp \in \bbN}$ such that there exists $b \in \{0,1\}$ such that $\Pr[Q_\secp(\ket{\psi_\secp}) = b] \geq 1-\nu(\secp)$, it holds that $$\Pr[\Eval(\widetilde{Q},\ket{\psi_\secp}^{\otimes k(\secp)}) = b : \widetilde{Q} \sample \obf(1^\secp,Q_\secp)] \geq 1-\mu(\secp).$$
\end{definition}
\begin{remark}
The correctness guarantee stated above is a weakening of standard obfuscation correctness in two ways. First, it only guarantees correctness for inputs that either accept or reject with high probability. Second, it requires the evaluator to possess multiple copies of the quantum input. However, note that for the class of pseudo-deterministic quantum circuits on classical inputs (\cref{def: deterministic}), where each input $x$ is mapped to a particular classical output $y = Q(x)$ with overwhelming probability, the above correctness guarantee is standard. Thus, this null-iO definition can be thought of as standard null-iO for psuedo-deterministic quantum circuits (strictly generalizing null-iO for classical circuits) with an additional correctness guarantee that holds when considering certain quantum inputs.
\end{remark}
We define the notion of security below.
\begin{definition}[Security]
A null obfuscator $(\obf,\Eval)$ is secure if for any negligible function $\nu$ and polynomial-size sequences of quantum circuits $\{U_{0,\secp}\}_{\secp \in \bbN}, \{U_{1,\secp}\}_{\secp \in \bbN}$ such that for all inputs $\{\ket{\psi_\secp}\}_{\secp \in \bbN}$, $\Pr[Q_{0,\secp}(\ket{\psi_\secp}) = 0] \geq 1-\nu(\secp)$ and $\Pr[Q_{1,\secp}(\ket{\psi_\secp}) = 0] \geq 1-\nu(\secp)$, it holds that $$\obf(1^\secp,Q_{0,\secp}) \approx_c \obf(1^\secp,Q_{1,\secp}).$$

\end{definition}

\subsection{Construction} 

Let $(\KeyGen,\TKeyGen,\SKeyGen,\Prove,\Verify,\TVerify)$ be a trapdoor dual-mode CVQC scheme, let $\obf$ be a post-quantum indistinguishability obfuscator, and let $\QFHE$ be a quantum fully-homomorphic encryption scheme. Our construction is presented in Figure~\ref{fig:null-iOnew}. An alternative construction assuming classical VBB obfuscation is given in Appendix~\ref{sec:vbb}.

\protocol{Null-iO for Quantum Circuits}{A null obfuscation scheme for quantum circuits.}{fig:null-iOnew}
{
\begin{itemize}
\item $\obf(1^\secp,Q)$: 
\begin{itemize}
\item Sample $(\pk,\sk) \sample \QFHE.\Gen(1^\secp)$.
\item Sample $(\params, r) \sample \KeyGen(1^\secp,Q)$ and $\ct_\params \gets \QFHE.\Enc(\pk,\params)$.
\item Let $\mathbf{C}[Q,\sk,r](\cdot)$ be the circuit that takes as input $\ct_\pi$ and outputs $$b = \Verify(Q,\QFHE.\Dec(\sk,\ct_\pi),r).$$
\item Compute $\widetilde{\mathbf{C}} \sample \obf(1^\secp,\mathbf{C}[Q,\sk,r])$.
\item Output $\widetilde{\mathbf{Q}} = (\ct_\params,\widetilde{\mathbf{C}})$.
\end{itemize}
\item $\Eval(\widetilde{\mathbf{Q}},\ket{\psi})$: Compute $\ct_\pi \sample \QFHE.\Eval(\pk,\Prove(\cdot,\ket{\psi}),\ct_\params)$ and output $\widetilde{\mathbf{C}}(\ct_\pi)$.
\end{itemize}
}

\paragraph{Analysis.} Correctness of the scheme follows immediately from correctness of the CVQC protocol, the post-quantum iO, and the QFHE. We show that our scheme is secure in an indistinguishability sense.

\begin{theorem}[Security]
Let $(\KeyGen,\TKeyGen,\SKeyGen,\Prove,\Verify,\TVerify)$ be a trapdoor dual-mode CVQC scheme, let $\obf$ be a post-quantum indistinguishability obfuscator, and let $\QFHE$ be a quantum fully-homomorphic encryption scheme. Then the protocol in Figure~\ref{fig:null-iOnew} is a secure quantum null-iO.
\end{theorem}

\begin{proof}
The proof proceeds by defining a series of hybrid distributions for the computation of the obfuscated circuit, moving from null circuit $Q_0$ to null circuit $Q_1$.
\begin{itemize}
    \item Hybrid $\mathcal{H}_0$: This is the honestly computed obfuscated circuit $(\ct_\params, \widetilde{\mathbf{C}})$ for $Q_0$.
    \item Hybrid $\mathcal{H}_1$: In this hybrid, we sample the public parameters of the CVQC scheme in trapdoor mode $(\params_\td, r_\td, \td) \sample \TKeyGen(1^\lambda, Q_0)$. This modification is indistinguishable by the setup indistinguishability of the dual-mode CVQC protocol.
    \item Hybrid $\mathcal{H}_2$: In this hybrid, we obfuscate the circuit $\mathbf{C}^\ast[Q_0,\sk,\td]$ that uses the algorithm $\TVerify(Q_0,\pi,\td)$ rather than $\Verify(Q_0,\pi,r_\td)$. By the verification equivalence of the CVQC protocol, the circuits $\mathbf{C}[Q_0,\sk,r_\td]$ and $\mathbf{C}^\ast[Q_0,\sk,\td]$ are functionally equivalent and therefore the indistinguishability follows from the security of $\obf$.
    \item $\Hyb_3:$ In this hybrid, we generate $(\params_\simgen,\td_\simgen) \sample \SKeyGen(1^\secp,Q_0)$. This is indistinguishable from $\Hyb_2$ due to the dual-mode property of the CVQC, since $Q_0$ is a null circuit.
    
    \item $\Hyb_4:$ In this hybrid, we let $\widetilde{\mathbf{C}}$ be an obfuscation of an always rejecting circuit, which is functionally equivalent to the obfuscated circuit in $\Hyb_2$ due to the dual-mode property of the CVQC.
    
    \item $\Hyb_5:$ In this hybrid, we sample $(\params_\simgen,\td_\simgen) \sample \SKeyGen(1^\secp,Q_1)$. This is indistinguishable from $\Hyb_4$ due to the semantic security of $\QFHE$, since $\params_\simgen$ is encrypted under $\QFHE$, and $\sk,\td_\simgen$ are not needed to produce the obfuscation.
    \item $\Hyb_6:$ We revert the change done in $\Hyb_4$, except that we obfuscate the circuit $\mathbf{C}^\ast[Q_1,\sk,\td_\simgen]$.
    \item $\Hyb_7:$ We revert the change done in $\Hyb_3$, except that we sample $(\params_\td, r_\td, \td) \sample \TKeyGen(1^\lambda, Q_1)$.
    \item $\Hyb_8:$ We revert the change done in $\Hyb_2$, except that we obfuscate the circuit $\mathbf{C}[Q_1,\sk,r_\td]$.
    \item $\Hyb_9:$ We revert the change done in $\Hyb_1$, except that we sample $(\params, r) \sample \KeyGen(1^\secp,Q_1)$.
    

\end{itemize}
The proof is concluded by observing that the last hybrid is identical to the correctly computed obfuscation of $Q_1$.
\end{proof}

\subsection{Witness Encryption for QMA}\label{sec:we}

Recall that a language $\lang = (\lang_{\text{yes}}, \lang_{\text{no}})$  in QMA is defined by a tuple $(\cV,p,\alpha,\beta)$, where $p$ is a polynomial, $\cV = \{V_\secp\}_{\secp \in \mathbb{N}}$ is a uniformly generated family of circuits such that for every $\secp$, $V_\secp$ takes as input a string $x \in \{0,1\}^\secp$ and a quantum state $\ket{\psi}$ on $p(\secp)$ qubits and returns a single bit, and $\alpha,\beta : \mathbb{N} \to [0,1]$ are such that $\alpha(\secp) - \beta(\secp) \geq 1/p(\secp)$. The language is then defined as follows.

\begin{itemize}
    \item For all $x \in \lang_{\text{yes}}$ of length $\secp$, there exists a quantum state $\ket{\psi}$ of size at most $p(\secp)$ such that the probability that $V_\secp$ accepts $(x, \ket{\psi})$ is at least $\alpha(\secp)$. We denote the (possibly infinite) set of quantum witnesses that make $V_\secp$ accept $x$ by $\relation(x)$.
    \item For all $x \in \lang_{\text{no}}$ of length $\secp$, and all quantum states $\ket{\psi}$ of size at most $p(\secp)$, it holds that $V_\secp$ accepts on input $(x,\ket{\psi})$ with probability at most $\beta(\secp)$.
\end{itemize}

We now recall the definition of witness encryption~\cite{STOC:GGSW13}, and adapt it to the quantum setting. Note that we define encryption only with respect to classical messages. This is without loss of generality, since one can encode a quantum state with the quantum one-time pad~\cite{QOTP} and use the witness encryption to encrypt the corresponding (classical) one-time pad keys.

\begin{definition}[Witness Encryption for QMA]
A witness encryption $(\WE.\Enc, \allowbreak\WE.\Dec)$ for a language $\lang\in\QMA$ with relation $\relation$ consists of the following efficient algorithms.
\begin{itemize}
    \item $\WE.\Enc(1^\lambda, x, m)$: On input the security parameter $1^\lambda$, a statement $x$, and a message $m\in\{0,1\}$, the encryption algorithm returns a ciphertext $c$.
    \item $\WE.\Dec(x, c, \ket{\psi})$: On input a statement $x$, a ciphertext $c$, and a quantum state $\ket{\psi}$, the decryption algorithm returns a message $m$ or $\bot$.
\end{itemize}
\end{definition}
We define correctness below. 

\begin{definition}[Correctness]
A witness encryption $(\WE.\Enc, \WE.\Dec)$ for a language $\lang\in\QMA$ is correct if there exists a negligible function $\nu(\secp)$ and a polynomial $k(\secp)$ such that for all $m\in\{0,1\}$, all polynomial-length sequences of instances $\{x_\secp\}_{\secp \in \bbN}$ and witnesses $\{\ket{\psi_\secp}\}_{\secp \in \bbN}$ where each $x_\secp \in \cL_\yes$ and $\ket{\psi_\secp}\in\relation(x_\secp)$, it holds that
\[
\Pr\left[\WE.\Dec(x_\secp,\WE.\Enc(1^\lambda, x_\secp, m), \ket{\psi_\secp}^{\otimes k(\lambda)}) = m\right] = 1 -\nu(\lambda).
\]
\end{definition}
Finally we recall the definition of security against quantum algorithms.
\begin{definition}[Security]
A witness encryption $(\WE.\Enc, \WE.\Dec)$ for a language $\lang\in\QMA$ is secure if for all polynomial-length sequences of instances $\{x_\secp\}_{\secp \in \bbN}$ where each $x_\secp \in \cL_\no$, it holds that
\[
\WE.\Enc(1^\lambda, x_\secp, 0) \approx_c \WE.\Enc(1^\lambda, x_\secp, 1).
\]
\end{definition}

\begin{lemma}
Assuming null-iO for quantum circuits satisfying Definition~\ref{def: null-iO}, there exists witness encryption for QMA.
\end{lemma}

\begin{proof}
$\WE.\Enc(1^\secp,x,m)$ will simply output $c \sample \obf(1^\secp,Q[x,m])$, where $Q[x,m]$ takes as input $\ket{\psi}$, runs the QMA verification procedure for instance $x$ and witness $\ket{\psi}$, and then outputs $m$ if verification accepts, and otherwise outputs $\bot$. $\WE.\Dec(x,c,\ket{\psi}^{\otimes k(\secp)})$ will take polynomially many copies of the witness $\ket{\psi}$ and run $\Eval(c,\ket{\psi}^{\otimes k(\secp)})$ to produce either $m$ or $\bot$. Assuming that the QMA language $\cL$ is such that $\alpha(\secp) = 1-\negl$ and $\beta(\secp) = \negl$ (which is without loss of generality by applying standard QMA amplification), correctness and security of the witness encryption scheme follow immediately from correctness and security of the null-iO. In particular, for $x \in \cL_\no$, $Q[x,m]$ is a null circuit, implying that $\obf(1^\secp,Q[x,m]) \approx_c \obf(1^\secp,Q[x,0])$.
\end{proof}

\section{Non-Interactive Zero-Knowledge for QMA}

In the following we show how a witness encryption scheme for QMA with classical encryption allows us to obtain a non-interactive zero-knowledge (NIZK) argument for QMA. Before describing our scheme, we introduce the necessary building blocks that we are going to use.

\subsection{Definition}

We recall the definition of NIZK for QMA.

\begin{definition}[NIZK Argument]
A NIZK argument $(\NIZK.\Setup,\allowbreak \NIZK.\Prove,\allowbreak \NIZK.\Verify)$ for a language $\lang\in\QMA$ with relation $\relation$ consists of the following efficient algorithms.
\begin{itemize}
    \item $\NIZK.\Setup(1^\lambda)$: On input the security parameter $1^\lambda$, the setup returns a common reference string $\crs$.
    \item $\NIZK.\Prove(\crs, \ket{\psi}^{\otimes k(\lambda)}, x)$: On input a common reference string $\crs$, $k(\lambda)$ copies of the witness $\ket{\psi}$, and a statement $x$, the proving algorithm returns a proof $\pi$.
    \item $\NIZK.\Verify(\crs, \pi, x)$: On input a common reference string $\crs$, a proof $\pi$, and a statement $x$, the verification algorithm returns a bit $\{0,1\}$.
\end{itemize}
\end{definition}
We defined correctness below.
\begin{definition}[Correctness]
A NIZK argument $(\NIZK.\Setup, \NIZK.\Prove, \allowbreak \NIZK.\Verify)$ is correct if there exists a negligible function $\nu$ such that for all $\lambda\in\mathbb{N}$, all $x\in\lang_\yes$, and all $\ket{\psi}\in\relation(x)$ it holds that
\[
\Pr\left[\NIZK.\Verify(\crs, \NIZK.\Prove(\crs, \ket{\psi}^{\otimes k(\secp)}, x), x) = 1\right] = 1 -\nu(\lambda)
\]
where $\crs \sample \NIZK.\Setup(1^\lambda)$.
\end{definition}
Next, we define (non-adaptive) computational soundness.
\begin{definition}[Computational Soundness]
A NIZK argument $(\NIZK.\Setup, \allowbreak\NIZK.\Prove, \allowbreak \NIZK.\Verify)$ is computationally sound if
there exist a negligible function $\nu$ such that for all non-uniform  QPT adversaries with quantum advice 
 $\adv = \{\adv_\lambda, \rho_\lambda\}_{\lambda \in \mathbb{N}}$ and all $x^\ast\in\lang_\no$, it holds that
 \[
 \Pr\left[\NIZK.\Verify(\crs, \adv_\secp(\crs, x^\ast; \rho_\secp), x^\ast) = 1\right] = \nu(\lambda)
 \]
 where $\crs \sample \NIZK.\Setup(1^\lambda)$.
\end{definition}
In the following we present the notion of (statistical) zero-knowledge.
\begin{definition}[Statistical Zero-Knowledge]
A NIZK argument $(\NIZK.\Setup,\allowbreak \NIZK.\Prove, \allowbreak \NIZK.\Verify)$ is statistically zero-knowledge if here exists a simulator $\mathsf{Sim}$ such that for $\lambda \in\mathbb{N}$, all $r\in\{0,1\}^\lambda$, all statements $x\in\lang_\yes$, and all witnesses $\ket{\psi}\in\relation(x)$,
it holds that
\[
\mathsf{Sim}(1^\lambda, x, r) \approx_s \NIZK.\Prove(\crs, \ket{\psi}^{\otimes k(\lambda)}, x)
\]
where $\crs = \NIZK.\Setup(1^\lambda; r)$.
\end{definition}

\subsection{Construction}

We describe in the following our NIZK argument system for any language $\lang\in\QMA$ with relation $\relation$. We assume the existence of a witness encryption $(\WE.\Enc, \allowbreak\WE.\Dec)$ with classical encryption for the same language $\lang$, a puncturable PRF $(\PRF.\Gen, \PRF.\Puncture, \PRF.\Eval)$, a one-way function $\OWF$, and an indistinguishability obfuscator $\obf$ for classical polynomial-size circuits. Our NIZK argument system $(\NIZK.\Setup, \allowbreak\NIZK.\Prove, \allowbreak\NIZK.\Verify)$ is presented in Figure~\ref{fig:nizk}.

\protocol{NIZK for QMA}{A publicly-verifiable NIZK argument for QMA}{fig:nizk}
{
\begin{itemize}
    \item $\NIZK.\Setup(1^\lambda)$:
    \begin{itemize}
        \item Sample two keys $k_0 \sample \PRF.\Gen(1^\lambda)$ and $k_1 \sample \PRF.\Gen(1^\lambda)$.
        \item Compute the obfuscation $\widetilde{\textbf{P}} \sample \obf(1^\lambda, \textbf{P})$ where $\textbf{P}$ is the circuit that, on input some statement $x$, returns $\WE.\Enc(1^\lambda, x, \PRF.\Eval(k_0, x);\PRF.\Eval(k_1, x))$. The circuit $\textbf{P}$ is padded to the maximum size of $\textbf{P}^\ast$ (defined in the proof of Theorem~\ref{thm:nizk_sound}).
        \item Compute the obfuscation $\widetilde{\textbf{V}} \sample \obf(1^\lambda, \textbf{V})$ where $\textbf{V}$ is the circuit that, on input some statement $x$ and a string $y$, returns $1$ if and only if $\OWF(\PRF.\Eval(k_0, x)) = \OWF(y)$. The circuit $\textbf{V}$ is padded to the maximum size of $\textbf{V}^\ast$ (defined in the proof of Theorem~\ref{thm:nizk_sound}).
        \item Return $\crs = (\widetilde{\textbf{P}}, \widetilde{\textbf{V}})$.
    \end{itemize}
    \item $\NIZK.\Prove(\crs, \ket{\psi}^{\otimes k(\lambda)}, x)$: 
    \begin{itemize}
    \item Compute $c = \widetilde{\textbf{P}}(x)$.
    \item Return $\pi = \WE.\Dec(x,c,\ket{\psi}^{\otimes k(\lambda)})$.
    \end{itemize}
    \item $\NIZK.\Verify(\crs, \pi, x)$: 
    \begin{itemize}
    \item Return $\widetilde{\textbf{V}}(x, \pi)$.
    \end{itemize} 
\end{itemize}
}

\paragraph{Correctness.}
It is easy to see that the scheme is correct, i.e.\ true statements correctly verify, except with negligible probability over the randomness imposed by the evaluation of the $\WE.\Dec$ algorithm. 

\paragraph{Soundness.}
Next we show that the scheme satisfies (non-adaptive) computational soundness.
\begin{theorem}[Soundness]\label{thm:nizk_sound}
Let $(\WE.\Enc, \WE.\Dec)$ be a witness encryption, let $(\PRF.\Gen, \PRF.\Puncture,\allowbreak \PRF.\Eval)$ be a puncturable PRF, let $\OWF$ be a one-way function, and let $\obf$ be an indistinguishability obfuscator. Then the scheme in Figure~\ref{fig:nizk} is computationally sound.
\end{theorem}

\begin{proof}
The proof proceeds by defining a series of hybrid distributions for the computation of the $\crs$ that we argue to be computationally indistinguishable from each other. In the last hybrid, the probability that any prover can cause the verifier to accept some $x^\ast \in\lang_\no$ will be negligible.
\begin{itemize}
    \item Hybrid $\mathcal{H}_0$: This is the original distribution where the $\crs$ is sampled from $\crs \sample \NIZK.\Setup$.
    \item Hybrid $\mathcal{H}_1$: In this hybrid we compute $\widetilde{\textbf{P}} \sample \obf(1^\lambda, \textbf{P}_1)$ where $\textbf{P}_1$ is the circuit that on input some statement $x$, checks whether $x = x^\ast$. If that is the case, then it returns the ciphertext $$c = \WE.\Enc(1^\lambda, x^\ast, \PRF.\Eval(k_0, x^\ast);\PRF.\Eval(k_1, x^\ast)).$$ Otherwise compute $c = \WE.\Enc(1^\lambda, x, \PRF.\Eval(k_0, x);\PRF.\Eval(k_{1, x^\ast}, x))$, where $k_{1, x^\ast} \sample \PRF.\Puncture(k_1, x^\ast)$.
    
    Note that the circuits $\textbf{P}$ and $\textbf{P}_1$ have different representations but are functionally equivalent. Thus, $\mathcal{H}_0$ and $\mathcal{H}_1$ are computationally indistinguishable by the security of the obfuscator $\obf$.
    
    \item Hybrid $\mathcal{H}_2$: In this hybrid we compute $\widetilde{\textbf{P}} \sample \obf(1^\lambda, \textbf{P}_2)$ where $\textbf{P}_2$ is defined as $\textbf{P}_1$ except that if $x = x^\ast$, then it returns the ciphertext $$c = \WE.\Enc(1^\lambda, x^\ast, \PRF.\Eval(k_0, x^\ast);u)$$ where $u \sample \{0,1\}^\lambda$.
    
    The indistinguishability $\mathcal{H}_1\approx_c\mathcal{H}_2$ follows by the pseudorandomness of the puncturable PRF.
    
    \item Hybrid $\mathcal{H}_3$: Here we compute $\widetilde{\textbf{P}} \sample \obf(1^\lambda, \textbf{P}_3)$ where $\textbf{P}_3$ is defined as $\textbf{P}_2$ except that if $x \neq x^\ast$, then it returns the ciphertext $$c = \WE.\Enc(1^\lambda, x, \PRF.\Eval(k_{0, x^\ast}, x);\PRF.\Eval(k_{1, x^\ast}, x))$$ where $k_{0, x^\ast} \sample \PRF.\Puncture(k_0, x^\ast)$.
    
    By the correctness of the puncturable PRF, the two circuits are functionally identical and therefore the computational indistinguishability follows from the security of $\obf$.
    
    \item Hybrid $\mathcal{H}_4$: In this hybrid we compute $\widetilde{\textbf{P}} \sample \obf(1^\lambda, \textbf{P}^\ast)$ where $\textbf{P}^\ast$ is defined as $\textbf{P}_3$ except that if $x = x^\ast$, then it returns the ciphertext $$c = \WE.\Enc(1^\lambda, x^\ast, 0^\lambda;u)$$ where $u \sample \{0,1\}^\lambda$.
    
    Recall that $x^\ast \in \lang_\no$ and thus indistinguishability between $\mathcal{H}_3$ and  $\mathcal{H}_4$ follows from the security of the witness encryption scheme.
    
    \item Hybrid $\mathcal{H}_5$: We now compute $\widetilde{\textbf{V}} \sample \obf(1^\lambda, \textbf{V}_1)$ where $\textbf{V}_1$ is the circuit that, on input a pair of strings $(x,y)$ checks whether $x = x^\ast$. If this is the case, then it returns $1$ if $\OWF(\PRF.\Eval(k_{0}, x^\ast)) = \OWF(y)$ and $0$ otherwise. If $x \neq x^\ast$ it returns $1$ if and only if $\OWF(\PRF.\Eval(k_{0, x^\ast}, x)) = \OWF(y)$ where $k_{0,x^\ast}$ is the punctured key.
    
    Observe that the circuits $\textbf{V}$ and $\textbf{V}_1$ are functionally equivalent and thus we can invoke the security of $\obf$ to show that $\mathcal{H}_4 \approx_c \mathcal{H}_5$.

    \item Hybrid $\mathcal{H}_6$: In this hybrid we compute $\widetilde{\textbf{V}} \sample \obf(1^\lambda, \textbf{V}_2)$ where $\textbf{V}_2$ is defined as $\textbf{V}_1$ except for the case where $x = x^\ast$. In this case the circuit returns $1$ if and only if $\OWF(r) = \OWF(y)$, where $r \sample \{0,1\}^\lambda$.
    
    The computational indistinguishability $\mathcal{H}_5 \approx_c \mathcal{H}_6$ follows from a reduction to the pseudorandomness of the puncturable PRF.
    
    \item Hybrid $\mathcal{H}_7$: In the final hybrid we compute $\widetilde{\textbf{V}} \sample \obf(1^\lambda, \textbf{V}^\ast)$ where $\textbf{V}^\ast$ is defined as $\textbf{V}_2$ except for the case where $x = x^\ast$. In this case the circuit returns $1$ if and only if $R = \OWF(y)$, where $R = \OWF(r)$, i.e.\ the image of the one-way function is hardwired in the circuit.
    
    Since the two circuits are functionally equivalent, we obatain that $\mathcal{H}_6 \approx_c \mathcal{H}_7$ by another invocation of the security of $\obf$.
\end{itemize}
Observe that causing the verifier to accept a proof $\pi$ for $x^\ast \in \lang_\no$ requires one to output a valid preimage of $R = \OWF(r)$, where $r$ is uniformly sampled. This is a contradiction to the one-wayness of $\OWF$ and concludes our proof.
\end{proof}

\paragraph{Zero Knowledge.}
We now show that the scheme satisfies a strong variant of statistical zero-knowledge. Namely, we show the existence of an efficient simulator whose output is statistically close to the output of the prover, for all valid choices of the common reference string.
\begin{theorem}[Zero Knowledge]
The scheme in Figure~\ref{fig:nizk} is statistically zero-knowledge.
\end{theorem}
\begin{proof}
The simulator computes $\crs$ as in the $\NIZK.\Setup$ algorithm and sets $\pi = \PRF.\Eval(k_0, x)$. This distribution is identical to the one induced by the honest algorithms, except when the $\WE.\Dec$ fails which happens only with negligible probability.
\end{proof}

\ifsubmission
\section{Attribute-Based Encryption for BQP}

In the following we present our construction of attribute-based encryption (ABE) for quantum functionalities.


\subsection{Definition}

We recall the definition of ABE. For convenience we consider the notion of \emph{ciphertext-policy} ABE where messages are encrypted with respect to circuits and keys are issued for attribute strings. If the class of circuits supported by the scheme is large enough, then one can switch to the complementary notion (i.e.\ \emph{key-policy} ABE) by encoding universal (quantum) circuits. We also consider without loss of generality an ABE that encrypts a single (classical) bit of information.

\begin{definition}[Attribute-Based Encryption for BQP]
An ABE scheme for BQP $(\ABE.\Gen,\allowbreak\ABE.\Enc,\allowbreak\ABE.\KeyGen,\allowbreak\ABE.\Dec)$ consists of the following efficient algorithms.
\begin{itemize}
    \item $\ABE.\Gen(1^\secp,1^\ell)$: On input the security parameter $1^\secp$ and the length $\ell$ of attributes, the parameters generation algorithm outputs a master public key $\mpk$, and master secret key $\msk$.
    \item $\ABE.\Enc(\mpk,Q,m)$: On input the master public key $\mpk$, a quantum circuit $Q$ (implementing a BQP language), and a message $m$, the encryption algorithm outputs a ciphertext $\ct_Q$.
    \item $\ABE.\KeyGen(\msk,x)$: On input the master secret key $\msk$ and an attribute $x$, the key generation algorithm outputs a secret key $\sk_x$.
    \item $\ABE.\Dec(\sk_x,\ct_Q)$: On input a secret key $\sk_x$ and a ciphertext $\ct_Q$, the decryption algorithms either outputs a message $m$ or $\bot$. 
\end{itemize}
\end{definition}
Throughout the rest of this work, we always assume that the ciphertexts also contain a description of the corresponding unitary $Q$ and that the keys also contain a description of the corresponding attribute $x$. We define correctness below.

\begin{definition}[Correctness]
An ABE scheme $(\ABE.\Gen,\ABE.\Enc,\allowbreak\ABE.\KeyGen,\allowbreak\ABE.\Dec)$ is correct if for all negligible functions $\nu$, there exists a negligible function $\mu$ such that for any $\secp \in \mathbb{N}, \ell \in \mathbb{N}, m \in \{0,1\}, x \in \{0,1\}^\ell$, and any quantum circuit $Q$ on $\ell$ input bits such that $\Pr[Q(x) = 1] = 1-\nu(\secp)$, it holds that
    $$\Pr\left[\ABE.\Dec\left(\ABE.\KeyGen(\msk,x),\ABE.\Enc(\mpk,Q,m)\right) = m \right] = 1 - \mu(\secp)$$
    where $(\mpk,\msk) \sample \ABE.\Gen(1^\secp,1^\ell)$.
\end{definition}
Finally we define the notion of security for ABE. We consider the selective notion of security, where the quantum circuit associated with the challenge ciphertext is known ahead of time. It is well known that this can be generically upgraded to the stronger notion of adaptive security (via complexity leveraging), although at the cost of an exponential decrease in the quality of the reduction.
\begin{definition}[Security]
An ABE scheme $(\ABE.\Gen,\ABE.\Enc,\ABE.\KeyGen,\allowbreak\ABE.\Dec)$ is secure if there exists a pair of negligible functions $\nu$  and $\mu$ such that for all $\lambda \in \mathbb{N}$, all quantum circuits $Q^\ast$, and all admissible non-uniform
 QPT distinguishers with quantum advice 
 $\adv = \{\adv_\lambda, \rho_\lambda\}_{\lambda \in \mathbb{N}}$, it holds that
 $$\begin{array}{r}
 \Pr\left[b =\cA_\secp(\ct_{Q^\ast}, \mpk; \rho_\secp)^{\ABE.\KeyGen(\msk,\cdot)} : \begin{array}{l}(\mpk,\msk) \sample \ABE.\Gen(1^\secp,1^\ell) \\ b \sample \{0,1\} \\ \ct_{Q^\ast} \sample \ABE.\Enc(\mpk,Q^\ast,b)\end{array}\right] = 1/2 + \mu(\secp)\end{array}$$
where $\adv$ is admissible if each query $x$ to $\ABE.\KeyGen(\msk,\cdot)$ is such that $\Pr[Q^\ast(x) = 1] \leq \nu(\secp)$.
\end{definition}


\subsection{Construction}

We are now ready the present our construction of ABE for BQP. We assume the existence of a pseudorandom generator $\PRG:\{0,1\}^\lambda \to \{0,1\}^{\ell\cdot\lambda}$, a puncturable PRF $(\PRF.\Gen, \allowbreak\PRF.\Puncture, \allowbreak\PRF.\Eval)$, and an indistinguishability obfuscator $\obf$ for classical circuits, all with sub-exponential security. We additionally assume the existence of a sub-exponentially secure witness encryption scheme $(\WE.\Enc, \WE.\Dec)$ for BQP with a classical encryption algorithm, which is implied by the scheme shown in Section~\ref{sec:we}. Our scheme $(\ABE.\Gen,\allowbreak\ABE.\Enc,\allowbreak \ABE.\KeyGen,\allowbreak \ABE.\Dec)$ is described in Figure~\ref{fig:abe}. 

\protocol{ABE for BQP}{An attribute-based encryption scheme for BQP}{fig:abe}
{
\begin{itemize}
    \item $\ABE.\Gen(1^\secp,1^\ell)$:
    \begin{itemize}
        \item Sample a key $k \sample \PRF.\Gen(1^\lambda)$.
        \item Compute the obfuscation $\widetilde{\textbf{P}} \sample \obf(1^\lambda, \textbf{P})$ where $\textbf{P}$ is the circuit that, on input some attribute $x\in\{0,1\}^\ell$ and a string $s\in\{0,1\}^\lambda$, returns $1$ if and only if $\PRG(s) = \PRG(\PRF.\Eval(k, x))$. The circuit $\textbf{P}[k]$ is padded to the maximum size of $\textbf{P}^\ast$ (defined in the proof of Theorem~\ref{thm:abesec}).
        \item Return $\msk = k$ and $\mpk = \widetilde{\textbf{P}}$.
    \end{itemize}
    \item $\ABE.\Enc(\mpk,Q,m)$:
    \begin{itemize}
        \item Sample a key $r \sample \PRF.\Gen(1^\lambda)$.
        \item Compute the obfuscation $\widetilde{\textbf{E}} \sample \obf(1^\lambda, \textbf{E})$ where $\textbf{E}$ is the circuit that, on input some attribute $x\in\{0,1\}^\ell$ and a string $s\in\{0,1\}^\lambda$,
        checks whether $\widetilde{\textbf{P}}(x,s)=1$ and returns $\WE.\Enc(1^\lambda, (Q,x), m; \PRF.\Eval(r,x))$ if this is the case. The circuit $\textbf{E}[m,r]$ is padded to the maximum size of $\textbf{E}^\ast$ (defined in the proof of Theorem~\ref{thm:abesec}).
        \item Return $\widetilde{\textbf{E}}$. 
    \end{itemize}
    \item $\ABE.\KeyGen(\msk,x)$:
    \begin{itemize}
        \item Return $\PRF.\Eval(k, x)$.
    \end{itemize}
    \item $\ABE.\Dec(\sk_x,\ct_Q)$:
    \begin{itemize}
        \item Parse $\ct_Q$ as $\widetilde{\textbf{E}}$ and compute $c = \widetilde{\textbf{E}}(x, \sk_x)$.
        \item Return $\WE.\Dec((Q,x),c)$.\footnote{Note that $\WE.\Dec$ does not need to take a third input (the witness) since the statement is in BQP.}
    \end{itemize}
\end{itemize}
}

\paragraph{Correctness.} The statistical correctness of the scheme follows immediately from the correctness of the underlying building blocks. In particular, note that all the circuits that we obfuscate are entirely classical and thus obfuscation for classical circuits suffices.

\paragraph{Security.} We analyze the (selective) security of our construction in the following. We remark that the proof can be upgraded to the stronger notion of adaptive security with additional complexity leveraging.
\begin{theorem}[Security]\label{thm:abesec}
Let $\PRG$ be a pseudorandom generator, let $(\PRF.\Gen,\allowbreak \PRF.\Puncture, \PRF.\Eval)$ be a puncturabl PRF, let $\obf$ be an indistinguishability obfuscator for classical circuits, and let $(\WE.\Enc, \WE.\Dec)$ be a witness encryption scheme, all with sub-exponential security (in the attribute size $\ell$). Then the scheme in Figure~\ref{fig:abe} is secure.
\end{theorem}

\begin{proof}
The proof proceeds by defining an exponentially long series of hybrids, iterating over all possible attributes $x\in\{0,1\}^\ell$. More specifically, starting from hybrid $\mathcal{H}_0$ (the original experiment with the bit $b$ fixed to $b=0$) we define, for each $i \in \{0,1\}^\ell$, a different sequence of hybrids and we argue about the indistinguishability of neighbouring distributions. As mentioned earlier, we assume all primitives we use are sub-exponentially secure, that is, there exists an $\epsilon > 0$ such that no efficient adversary can break the primitive with probability better than $2^{-\secp^\epsilon}$. Thus, we can set the security parameter for each to be at least $\ell^c$ for some $c > 1/\epsilon$, ensuring that efficient adversaries have advantage $\negl / 2^\ell$. 

\begin{itemize}
    \item Hybrid $\mathcal{H}_{i,0}$: Defined the previous hybrid, except that we change the way we compute the challenge ciphertext. We begin by computing a punctured key $r_i \sample \PRF.\Puncture(r, i)$. Then we compute $\widetilde{\textbf{E}} \sample \obf(1^\lambda, \textbf{E}_1)$, where $\textbf{E}_1$ takes as input a pair $(x,s)$ and does the following.
    \begin{itemize}
        \item If $x < i$: Check whether $\widetilde{\textbf{P}}(x,s)=1$ and return $$\WE.\Enc(1^\lambda, (Q,x), 1; \PRF.\Eval(r_i,x))$$ if this is the case.
        \item If $x = i$: Check whether $\widetilde{\textbf{P}}(x,s)=1$ and return $$\WE.\Enc(1^\lambda, (Q,x), 0; \PRF.\Eval(r,x))$$ if this is the case.
        \item If $x > i$: Check whether $\widetilde{\textbf{P}}(x,s)=1$ and return $$\WE.\Enc(1^\lambda, (Q,x), 0; \PRF.\Eval(r_i,x))$$ if this is the case.
    \end{itemize}
    
    Note that, by the correctness of the puncturable PRF, the circuits $\textbf{E}$ and $\textbf{E}_1$ are functionally equivalent and therefore indistinguishability follows from the security of the classical obfuscator $\obf$.

    \item Hybrid $\mathcal{H}_{i,1}$: Defined the previous hybrid, except that we compute $\widetilde{\textbf{E}} \sample \obf(1^\lambda, \textbf{E}_2)$, where $\textbf{E}_2$ takes as input a pair $(x,s)$ and does the following.
    \begin{itemize}
        \item If $x < i$: Same as $\textbf{E}_1$.
        \item If $x = i$: Check whether $\widetilde{\textbf{P}}(x,s)=1$ and return $\WE.\Enc(1^\lambda, (Q,x), 0; \tilde{r})$ if this is the case, where $\tilde{r} \sample \{0,1\}^\lambda$.
        \item If $x > i$: Same as $\textbf{E}_1$.
    \end{itemize}
    
    Note that the two hybrids differ only in the definition of $\tilde{r}$, which is uniformly sampled in $\mathcal{H}_{i,1}$ and computed according to the puncturable PRF in $\mathcal{H}_{i,0}$. By the indistinguishability of the puncturable PRF, we have that the two distributions are computationally close.
    
    \item Hybrid $\mathcal{H}_{i,2}$: In this hybrid we check whether $Q^\ast(i) = 0$. If this is not the case, then we proceed as before. Otherwise, we compute $\widetilde{\textbf{E}} \sample \obf(1^\lambda, \textbf{E}_3)$, where $\textbf{E}_3$ takes as input a pair $(x,s)$ and does the following.
    \begin{itemize}
        \item If $x < i$: Same as $\textbf{E}_2$.
        \item If $x = i$: Check whether $\widetilde{\textbf{P}}(x,s)=1$ and return $\WE.\Enc(1^\lambda, (Q,x), 1; \tilde{r})$ if this is the case, where $\tilde{r} \sample \{0,1\}^\lambda$.
        \item If $x > i$: Same as $\textbf{E}_2$.
    \end{itemize}
    
    Note that we change the view of the adversary only if $Q^\ast(i) = 0$, which implies that the statement $(Q,i)$ is false. Thus indistinguishability follows from the security of the witness encryption scheme.

    \item Hybrid $\mathcal{H}_{i,3}$: This is defined as the previous one, except that we compute a punctured key $k_i \sample \PRF.\Puncture(k, i)$ and we modify the public parameters as follows. We obfuscate $\widetilde{\textbf{P}} \sample \obf(1^\lambda, \textbf{P}_1)$ where $\textbf{P}_1$ is the circuit that, on input some attribute $x\in\{0,1\}^\ell$ and a string $s\in\{0,1\}^\lambda$, does the following.
    \begin{itemize}
        \item If $x\neq i$: Return $1$ if and only if $\PRG(s) = \PRG(\PRF.\Eval(k_i, x))$.
        \item If $x=i$: Return $1$ if and only if $\PRG(s) = \PRG(\PRF.\Eval(k, i))$.
    \end{itemize}
    
    By the perfect correctness of the puncturable PRF, the two circuits are functionally equivalent and therefore the indistinguishability follows from the security of the obfuscator $\obf$.

    \item Hybrid $\mathcal{H}_{i,4}$: In this hybrid we compute $\widetilde{\textbf{P}} \sample \obf(1^\lambda, \textbf{P}_2)$ where $\textbf{P}_2$ is the circuit that, on input some attribute $x\in\{0,1\}^\ell$ and a string $s\in\{0,1\}^\lambda$, does the following.
    \begin{itemize}
        \item If $x\neq i$: Same as $\textbf{P}_1$.
        \item If $x=i$: Return $1$ if and only if $\PRG(s) = \PRG(\tilde{k})$, where $\tilde{k}\sample\{0,1\}^\lambda$.
    \end{itemize}   
    Additionally, we answer the query of the adversary to the key generation oracle with $\tilde{k}$, if queried on attribute $i$.

    Note that this hybrid is identical to the previous one, except that $\tilde{k}$ is sampled uniformly. By the security of the puncturable PRF, the two hybrids are computationally indistinguishable.
    
    \item Hybrid $\mathcal{H}_{i,5}$: Before sampling the public parameters, we check whether $Q^\ast(i) = 1$. If this is not the case, then we proceed as before. Otherwise we obfuscate $\widetilde{\textbf{P}} \sample \obf(1^\lambda, \textbf{P}_3)$ where $\textbf{P}_3$ is the circuit that, on input some attribute $x\in\{0,1\}^\ell$ and a string $s\in\{0,1\}^\lambda$, does the following.
    \begin{itemize}
        \item If $x\neq i$: Same as $\textbf{P}_2$.
        \item If $x=i$: Return $1$ if and only if $\PRG(s) = K$, where $K \sample\{0,1\}^{\lambda\cdot\ell}$.
    \end{itemize}   
    
    Note that if $Q^\ast(i) \neq 1$, then the distribution induced by this hybrid is identical to the previous one, so we only consider the case where $Q^\ast(i) = 1$. Observe that an admissible adversary never queries the key generation oracle on $i$. Thus, the key $\tilde{k}$ is not present in the view of the distinguisher. Indistinguishability follows from the pseudorandomness of $\PRG$.
    
    \item Hybrid $\mathcal{H}_{i,6}$: Here we again check whether $Q^\ast(i) = 1$. If this is not the case, then we proceed as before. Otherwise we compute $\widetilde{\textbf{P}} \sample \obf(1^\lambda, \textbf{P}^\ast)$ where $\textbf{P}^\ast$ is the circuit that, on input some attribute $x\in\{0,1\}^\ell$ and a string $s\in\{0,1\}^\lambda$, does the following.
    \begin{itemize}
        \item If $x\neq i$: Same as $\textbf{P}_3$.
        \item If $x=i$: Return $0$.
    \end{itemize}  

    Note that the programs $\textbf{P}_3$ and $\textbf{P}^\ast$ are identical except if $K$ falls within the range of $\PRG$. Since this happens only with negligible probability, then the two hybrids are computationally indistinguishable by the security of the obfuscator $\obf$.
    
    \item Hybrid $\mathcal{H}_{i,7}$: In this hybrid we check whether $Q^\ast(i) = 1$. If this is not the case, then we proceed as before. Otherwise, we compute the challenge ciphertext as $\widetilde{\textbf{E}} \sample \obf(1^\lambda, \textbf{E}^\ast)$, where $\textbf{E}^\ast$ takes as input a pair $(x,s)$ and does the following.
    \begin{itemize}
        \item If $x < i$: Same as $\textbf{E}_3$.
        \item If $x = i$: Check whether $\widetilde{\textbf{P}}(x,s)=1$ and return $\WE.\Enc(1^\lambda, (Q,x), 1; \tilde{r})$ if this is the case, where $\tilde{r} \sample \{0,1\}^\lambda$.
        \item If $x > i$: Same as $\textbf{E}_3$.
    \end{itemize}
    
    Observe that at this point $\widetilde{\textbf{P}}$ always returns $0$ whenever queried on $i$, and thus the programs $\textbf{E}_3$ and $\textbf{E}^\ast$ are functionally equivalent. Indistinguishability follows from the security of $\obf$.

    \item Hybrid $\mathcal{H}_{i,8}$: We revert the change done in $\mathcal{H}_{i,6}$.
    
    \item Hybrid $\mathcal{H}_{i,9}$: We revert the change done in $\mathcal{H}_{i,5}$.

    \item Hybrid $\mathcal{H}_{i,10}$: We revert the change done in $\mathcal{H}_{i,4}$.
    
    \item Hybrid $\mathcal{H}_{i,11}$: We revert the change done in $\mathcal{H}_{i,3}$.

    \item Hybrid $\mathcal{H}_{i,12}$: We revert the change done in $\mathcal{H}_{i,1}$.

    \item Hybrid $\mathcal{H}_{i,13}$: We revert the change done in $\mathcal{H}_{i,0}$.
\end{itemize}
We denote by $\mathcal{H}_1$ the last hybrid of the sequence $\mathcal{H}_{2^\ell, 13}$. Observe that such an hybrid is identical to the original experiment with the bit $b$ fixed to $b=1$. This concludes our proof.
\end{proof}

\begin{remark}
We remark that using null-iO instead of witness encryption in the above scheme would allow us to achieve a stronger security definition where the attributes are also hiding to the eyes of parties that have keys for non-accepting predicates (one-sided attribute hiding). In Appendix~\ref{sec:pe} we show an alternative transformation from any ABE scheme to one-sided attribute-hiding ABE, additionally assuming the post-quantum hardness of the LWE problem.
\end{remark}

\else
\section{ZAPR Arguments for QMA}

In the following we show a transformation to lift our NIZK argument for QMA to the setting where the common reference string can be sampled maliciously. Specifically we construct a two-message witness indistinguishable argument for QMA with public verifiability. Such an argument has been referred to in the literature as a ZAPR argument (i.e.\ a ZAP~\cite{FOCS:DwoNao00} where the first message may be sampled with private random coins that are not needed for verification). Before presenting our scheme, we introduce the necessary cryptographic machinery.

\subsection{Definition}

In the following we define the notion of statistical ZAPR arguments for QMA, although a similar definition applies (with minor modifications) to the case of NP.

\begin{definition}[ZAPR Argument]
A ZAPR argument $(\ZAPR.\Setup,\allowbreak \ZAPR.\Prove,\allowbreak \ZAPR.\Verify)$ for a language $\lang\in\QMA$ with relation $\relation$ consists of the following efficient algorithms.
\begin{itemize}
    \item $\ZAPR.\Setup(1^\lambda)$: On input the security parameter $1^\lambda$, the setup returns a common reference string $\crs$.
    \item $\ZAPR.\Prove(\crs, \ket{\psi}^{\otimes k(\lambda)}, x)$: On input a common reference string $\crs$, $k(\lambda)$ copies of the witness $\ket{\psi}$, and a statement $x$, the proving algorithm returns a proof $\pi$.
    \item $\ZAPR.\Verify(\crs, \pi, x)$: On input a common reference string $\crs$, a proof $\pi$, and a statement $x$, the verification algorithm returns a bit $\{0,1\}$.
\end{itemize}
\end{definition}
We define correctness below.
\begin{definition}[Correctness]
A ZAPR argument $(\ZAPR.\Setup, \ZAPR.\Prove, \allowbreak \ZAPR.\Verify)$ is correct if there exists a negligible function $\nu$ such that for all $\lambda\in\mathbb{N}$, all $x\in\lang_\yes$, and all $\ket{\psi}\in\relation(x)$ it holds that
\[
\Pr\left[\ZAPR.\Verify(\crs, \ZAPR.\Prove(\crs, \ket{\psi}^{\otimes k(\secp)}, x), x) = 1\right] = 1 -\nu(\lambda)
\]
where $\crs \sample \ZAPR.\Setup(1^\lambda)$.
\end{definition}
Next, we define computational soundness.
\begin{definition}[Computational Soundness]
A ZAPR argument $(\ZAPR.\Setup,\allowbreak \ZAPR.\Prove, \allowbreak \ZAPR.\Verify)$ is computationally sound if
there exist a negligible function $\nu$ such that for all non-uniform  QPT adversaries with quantum advice 
$\adv = \{\adv_\lambda, \rho_\lambda\}_{\lambda \in \mathbb{N}}$ and all $x^\ast\in\lang_\no$, it holds that
\[
\Pr\left[\ZAPR.\Verify(\crs, \adv_\secp(\crs, x^\ast; \rho_\secp), x^\ast) = 1\right] = \nu(\lambda)
 \]
 where $\crs \sample \ZAPR.\Setup(1^\lambda)$.
\end{definition}
In the following we present the notion of (statistical) witness indistinguishability.
\begin{definition}[Statistical Witness Indistinguishability]
A ZAPR argument $(\ZAPR.\Setup, \ZAPR.\Prove, \allowbreak \ZAPR.\Verify)$ is witness indistinguishable if for all $\lambda\in\mathbb{N}$, all $x\in\lang_\yes$, all pairs of witnesses $\ket{\psi_0}, \ket{\psi_1} \in\relation(x)$, and all common reference strings $\crs$, it holds that
\[
\ZAPR.\Prove(\crs, \ket{\psi_0}^{k(\lambda)}, x) \approx_s  \ZAPR.\Prove(\crs, \ket{\psi_1}^{k(\lambda)}, x).
\]
\end{definition}

\paragraph{Statistical ZAPs for NP.}
It was recently show in~\cite{EC:BFJKS20,EC:GJJM20} that statistical ZAPs for $\NP$ exist assuming the quasi-polynomial (quantum) hardness of the LWE problem.
\begin{lemma}[\cite{EC:BFJKS20,EC:GJJM20}]
Assuming the quantum quasi-polynomial hardness of the LWE problem, there exists a public coin ZAP for $\NP$
$(\ZAP.\Setup,\allowbreak \ZAP.\Prove, \ZAP.\Verify)$.
\end{lemma}

\subsection{Non-Interactive Witness-Indistinguishable Proofs for NP}

We recall the notion of non-interactive witness-indistinguishable (NIWI) proof for NP~\cite{C:BarOngVad03}.
\begin{definition}[NIWI Proof for NP]
A NIWI proof $(\NIWI.\Prove, \NIWI.\Verify)$ for a language $\lang\in\NP$ with relation $\relation$ consists of the following efficient algorithms.
\begin{itemize}
    \item $\NIWI.\Prove(1^\lambda,w, x)$: On input the security parameter $1^\lambda$, a witness $w$, and a statement $x$, the proving algorithm returns a proof $\pi$.
    \item $\NIWI.\Verify(\pi, x)$: On input a proof $\pi$, and a statement $x$, the verification algorithm returns a bit $\{0,1\}$.
\end{itemize}
\end{definition}
We defined the properties of interest below.
\begin{definition}[Correctness]
A NIWI proof $(\NIWI.\Prove, \NIWI.\Verify)$ is correct if for all $\lambda\in\mathbb{N}$, all $x\in\lang$, and all $w\in\relation(x)$ it holds that
\[
\Pr\left[\NIWI.\Verify(\NIWI.\Prove(1^\lambda, w, x), x) = 1\right] = 1.
\]
\end{definition}
We define statistical soundness.
\begin{definition}[Statistical Soundness]
A NIWI proof $(\NIWI.\Prove, \NIWI.\Verify)$  is statistically sound if
there exist a negligible function $\nu$ such that for all $x^\ast\notin\lang$ and all proofs $\pi^\ast$ it holds that
 \[
 \Pr\left[\NIWI.\Verify(\pi^\ast, x^\ast) = 1\right] = \nu(\lambda).
 \]
\end{definition}
Finally we define computational witness indistinguishability.
\begin{definition}[Computational Witness Indistinguishability]
A NIWI proof $(\NIWI.\Prove, \NIWI.\Verify)$ is witness indistinguishable if there exist a negligible function $\nu$ such that for all $\lambda\in\mathbb{N}$, all $x\in\lang$, and all pairs of witnesses $w_0, w_1\in\relation(x)$ it holds that
\[
\NIWI.\Prove(1^\lambda,w_0, x) \approx_c \NIWI.\Prove(1^\lambda,w_1, x).
\]
\end{definition}
NIWI proofs are known to exist under a variety of assumptions, but for the purpose of our work we only consider constructions that (plausibly) satisfy post-quantum security.
\begin{lemma}[\cite{TCC:BitPan15},\cite{TCC:BitPanWic16}]
Assuming the existence of post-quantum one-way functions and post-quantum sub-exponential indistinguishability obfuscation for classical circuits, there exists a post-quantum NIWI for $\NP$
$(\NIWI.\Prove, \allowbreak\NIWI.\Verify)$.
\end{lemma}

\subsection{Sometimes-Binding Statistically Hiding Commitments}

We introduce the notion of sometimes-binding statistically hiding (SBSH) commitments, a notion formally introduced in~\cite{EC:LomVaiWic20}.

\begin{definition}[SBSH Commitment]
An SBSH commitment scheme $(\SBSH.\Gen, \allowbreak\SBSH.\Key, \SBSH.\com)$ consists of the following efficient algorithms.
\begin{itemize}
    \item $\SBSH.\Gen(1^\lambda)$: On input the security parameter $1^\lambda$, the generation algorithm returns a partial commitment key $\ck_0$.
    \item $\SBSH.\Key(\ck_0)$: On input a partial key $\ck_0$, the key agreement algorithm returns the complement of the key $\ck_1$.
    \item $\SBSH.\com((\ck_0, \ck_1), m)$: On input a commitment key $(\ck_0, \ck_1)$ and a message $m$, the commitment algorithm returns a partial commitment key $\ck_1$ and a commitment $c$.
\end{itemize}
\end{definition}
The commitment must satisfy the notion of statistical hiding.
\begin{definition}[Statistical Hiding]
An SBSH commitment scheme $(\SBSH.\Gen,\allowbreak \SBSH.\Key, \SBSH.\com)$  is statistically hiding if for all $\lambda\in\mathbb{N}$, all partial keys $\ck_0$, and all pairs of messages $(m_0, m_1)$, it holds that
\[
(\ck_0,\ck_1,\SBSH.\com((\ck_0, \ck_1), m_0)) \approx_s (\ck_0,\ck_1,\SBSH.\com((\ck_0, \ck_1), m_1))
\]
where $\ck_1 \sample \SBSH.\Key(\ck_0)$.
\end{definition}
Next we define the notion of sometimes-binding for an SBSH commitment scheme. We define the set $\mathsf{Binding}$ as the set of all commitment keys $(\ck_0, \ck_1)$ such that any resulting commitment is perfectly binding. We present the definition of the property in the following.
\begin{definition}[Sometimes Binding]
An SBSH commitment scheme $(\SBSH.\Gen,\allowbreak \SBSH.\Key, \SBSH.\com)$ is $(\varepsilon, \delta)$-sometimes binding if there exists a negligible function $\nu$ such that for all $\lambda\in\mathbb{N}$ and all (stateful) QPT distinguishers $\adv = \{\adv_\lambda, \rho_\lambda\}_{\lambda \in \mathbb{N}}$, it holds that
\[
\Pr\left[\adv_\secp(\tau) =1 \land (\ck_0, \ck_1) \in \mathsf{Binding}\right] = \varepsilon(\lambda) \cdot \Pr\left[\adv_\secp(\tau) =1\right] - \delta(\lambda)\cdot\nu(\lambda) 
\]
where $\ck_0 \sample \SBSH.\Gen(1^\lambda)$ and $(\tau,\ck_1) \gets \adv_\secp(\ck_0,\rho_\secp)$.
\end{definition}
We also require the existence of a polynomial-time extractor $\SBSH.\Ext$ that, on input the random coins $r$ used in the $\SBSH.\Gen$ algorithm, extracts the committed message $m$ from the protocol transcript if $(\ck_0, \ck_1) \in \mathsf{Binding}$. The works of~\cite{EC:KalKhuSah18,EC:BFJKS20,EC:GJJM20} present constructions of SBSH commitment schemes (albeit with a slightly different syntax) for quasi-polynomial $(\varepsilon, \delta)$ assuming the quasi-polynomial hardness of LWE. Note that the extractor does not need to access the code of the adversary (not even as an oracle) and therefore it is well defined regardless on whether the adversary is classical or quantum.

\begin{lemma}[\cite{EC:KalKhuSah18,EC:BFJKS20,EC:GJJM20}]
Assuming the quantum quasi-polynomial hardness of the LWE problem, there exists an $(\varepsilon, \delta)$-sometimes binding SBSH commitment scheme $(\SBSH.\Gen, \SBSH.\Key, \SBSH.\com)$.
\end{lemma}

\subsection{Construction}

We are now in the position to present a formal description of our ZAPR argument system for any language $\lang\in\QMA$ with relation $\relation$. Let $\varepsilon$ be some fixed negligible function. We assume the existence of a NIZK for $\lang$ $(\NIZK.\Setup, \allowbreak\NIZK.\Prove, \NIZK.\Verify)$ with classical setup and classical verification, a one-way function $\OWF$, a statistical ZAP for NP $(\ZAP.\Setup, \ZAP.\Prove, \ZAP.\Verify)$, and a NIWI for NP $(\NIWI.\Prove, \NIWI.\Verify)$, all with quasi-polynomial security $\varepsilon(\lambda)^2 \cdot \nu(\lambda)$, for some negligible function $\nu$. Finally, we assume the existence of an SBSH commitment scheme $(\SBSH.\Gen, \SBSH.\Key, \SBSH.\com)$ with $(\varepsilon(\lambda), \varepsilon(\lambda)^2)$-sometimes binding. Our protocol is presented in Figure~\ref{fig:zapr}.

\protocol{ZAPR for QMA}{A publicly-verifiable ZAPR argument for QMA}{fig:zapr}
{
\begin{itemize}
   \item $\ZAPR.\Setup(1^\lambda)$:
   \begin{itemize}
       \item Sample two common reference strings for the NIZK  system $\crs_0 \sample \NIZK.\Setup(1^\lambda)$ and $\crs_1 \sample \NIZK.\Setup(1^\lambda)$.
       \item Sample two strings $(x_0, x_1) \sample\{0,1\}^{2\lambda}$ and compute the corresponding images $y_0 = \OWF(x_0)$ and $y_1 = \OWF(x_1)$.
       \item Compute the partial key of the SBSH commitment $\ck_0 \sample \SBSH.\Gen(1^\lambda)$ and the first message of the ZAP $\crs'' \sample \ZAP.\Setup(1^\lambda)$.
       \item Compute a NIWI proof $\pi'$ for the statement
       \begin{align*}
      \{ \left( \crs_0 \in \NIZK.\Setup \text{ AND }y_0 \in \OWF\right) \text{ OR } \left( \crs_1 \in \NIZK.\Setup \text{ AND }y_1 \in \OWF\right)\}.
       \end{align*}
       \item Return $\crs = (\crs_0, \crs_1, y_0, y_1, \ck_0, \crs'', \pi')$.
   \end{itemize}
  \item $\ZAPR.\Prove(\crs, \ket{\psi}^{\otimes 2k(\lambda)}, x)$:
   \begin{itemize}
    \item Verify $\pi'$ and abort if the verification does not succeed.
    \item Compute two NIZK proofs for the statement $x$, $\pi_0 \sample \NIZK.\Prove(\crs_0, \ket{\psi}^{\otimes k(\lambda)}, x)$ and $\pi_1 \sample \NIZK.\Prove(\crs_1, \ket{\psi}^{\otimes k(\lambda)}, x)$. If $\pi_0$ is a valid proof, then set $b=0$, else if $\pi_1$ is valid then set $b=1$. If neither of the two proofs is valid, then abort.
    \item Sample a partial key for the SBSH commitment $\ck_{1} \sample \SBSH.\Key(\ck_0)$.
    \item Compute two SBSH commitments $c_{\NIZK} \sample \SBSH.\com((\ck_0, \ck_{1}), (b, \pi_b))$ and $c_{\OWF} \sample \SBSH.\com((\ck_0, \ck_{1}), 0)$.
    \item Compute a ZAP proof $\pi''$ for the statement
    \[
    \left\{ 
    \begin{array}{l}
    c_{\NIZK} \in \SBSH.\com((\ck_0, \ck_{1}), (b, \pi_b))\text{ s.t. } \NIZK.\Verify(\crs_b, \pi_b, x) = 1 \\\text{OR } c_{\OWF} \in \SBSH.\com((\ck_0, \ck_{1}), x_b) \text{ s.t. } \OWF(x_b) = y_b
        \end{array}
\right\}.
    \]
   \item Return $\pi = (\ck_{1}, c_{\NIZK}, c_{\OWF}, \pi'')$.
   \end{itemize}
   \item $\ZAPR.\Verify(\crs, \pi, x)$: 
    \begin{itemize}
    \item Accept if and only if $\pi''$ verifies.
    \end{itemize} 
\end{itemize}
}
\paragraph{Correctness.}
The correctness of the scheme follows immediately from the correctness of the underlying primitives.

\paragraph{Soundness.} We show how to reduce the (non-adaptive) computational soundness of our protocol to the security of the corresponding cryptographic primitives.
\begin{theorem}[Soundness]
Let $(\SBSH.\Gen, \SBSH.\Key, \SBSH.\com)$ be an SBSH commitment with $(\varepsilon(\lambda), \varepsilon(\lambda)^2)$-sometimes binding. Let $(\NIZK.\Setup, \allowbreak\NIZK.\Prove, \allowbreak\NIZK.\Verify)$ be a NIZK for $\lang$, let $\OWF$ be a one-way function, let $(\ZAP.\Setup, \allowbreak\ZAP.\Prove, \allowbreak\ZAP.\Verify)$ be a ZAP for NP, and let $(\NIWI.\Prove, \allowbreak \NIWI.\Verify)$, all with negligible security in $\varepsilon(\lambda)^2$. Then the scheme in Figure~\ref{fig:zapr} is computationally sound.
\end{theorem}

\begin{proof}
The proof proceeds by showing a bound on the success probability of the prover of $\varepsilon(\lambda)$. Assume towards contradiction that 
\[
\Pr\left[\ZAPR.\Verify(\crs, \adv_\secp(\crs, x^\ast; \rho_\secp), x^\ast) = 1\right] \geq \varepsilon(\lambda).
\]
By the $(\varepsilon(\lambda), \varepsilon(\lambda)^2)$-sometimes binding property of the SBSH commitment scheme, we have that 
\[
\Pr\left[\ZAPR.\Verify(\crs, \adv_\secp(\crs, x^\ast; \rho_\secp), x^\ast) = 1\land (\ck_0, \ck_1)\in\mathsf{Binding}\right] \geq \varepsilon(\lambda)^2 \cdot(1- \nu(\lambda))
\]
for some negligible function $\nu(\lambda)$. We denote the outputs of the SBSH extractor by $r_\NIZK^\ast = \SBSH.\mathsf{Ext}(r,\allowbreak \ck_0, \ck_1, c_{\NIZK})$ and $r_\OWF^\ast = \SBSH.\mathsf{Ext}(r, \ck_0, \ck_1, c_{\OWF})$, where $r$ denotes the random coins used in the $\SBSH.\Gen$ algorithm. We now proceed by modifying the verification procedure to derive a contradiction.
\begin{itemize}
\item The verifier additionally checks whether $r_\NIZK^\ast$ or $r_\OWF^\ast$ define a bit $b\in\{0,1\}$, either by containing a valid proof $\pi_b$ with respect to $\crs_b$ or by containing a pre-image of $y_b$. If no such bit is defined, then the verifier aborts.

We claim that the probability that the adversary cheats and that the bit $b$ is well defined is at least $\varepsilon(\lambda)^2\cdot 1/\poly$. Assume towards contradiction that this is not the case. Then we know that with probability at least $\varepsilon(\lambda)^2 \cdot(1- \negl)$ the adversary successfully cheats and the ZAP proof $\pi''$ proves a false statement. This contradicts the quasi-polynomial soundness of the ZAP for NP. It follows that a well-defined $b$ is extracted by the proof with probability inverse polynomial in $\varepsilon(\lambda)^2$.

\item During the computation of the common reference string, the verifier samples a bit $b'$ and computes the NIWI proof using $x_{b'\oplus 1}$ and the random coins of the generation of $\crs_{b'\oplus 1}$ as a witness. In the verification algorithm, the verifier aborts if $b'\neq b$. 

We claim that, conditioned on the extraction being successful, the probabiliy that the verifier aborts is negligibly close to $1/2$. Assume the contrary, and consider the following reduction against the quasi-polynomial witness indistinguishability of the NIWI proof: The reduction performs the same operations of the verifier and, if the extraction is not successful, then it outputs a uniformly sampled bit, otherwise it outputs the extracted bit $b$. As argued above, the extraction succeeds with probability at least inverse polynomial in $\varepsilon(\lambda)^2$. Thus, the advantage of the reduction is also at least inverse polynomial in $\varepsilon(\lambda)^2$, which contradicts the security of the NIWI proof.

The bias on the output of the reduction is identical to the probability that $b' \neq b$, conditioned on the fact that the extraction is successful. As argued above, the extraction succeeds with probability at least inverse polynomial in $\varepsilon(\lambda)^2$. This is a contradiction to the quasi-polynomial soundness of the NIWI.
\end{itemize}
By the above analysis, it follows that, with probability inverse polynomial in $\varepsilon(\lambda)^2$, the variables $(r_\NIZK^\ast, r_\OWF^\ast)$ encode either (1) a valid NIZK proof $\pi_b$ against $\crs_b$ or (2) the pre-image of $y_b$. Note that the random coins used to sample $\crs_b$ and $y_b$ are not used by the verifier to compute the NIWI proof. Thus, case (1) contradicts the quasi-polynomial soundness of the NIZK argument, whereas (2) contradicts the quasi-polynomial security of the one-way function.
\end{proof}

\paragraph{Witness Indistinguishability.} In the following we show that our scheme satisfies statistical witness indistinguishability.

\begin{theorem}[Witness Indistinguishability]
The scheme in Figure~\ref{fig:zapr} is statistically witness indistinguishable.
\end{theorem}
\begin{proof}
The proof proceeds by defining a sequence of hybrid distributions that we show to be statistically indistinguishable.
\begin{itemize}
    \item Hybrid $\mathcal{H}_0$: This is the distribution with the proof being computed using $\ket{\psi_0}$, i.e.\ this is the distribution $\ZAPR.\Prove(\crs, \ket{\psi_0}^{\otimes 2k(\lambda)}, x)$.
    \item Hybrid $\mathcal{H}_1$: In this hybrid the algorithm checks inefficiently whether one of the two $(\crs_0, \crs_1)$ is correctly computed and whether the corresponding $(y_0, y_1)$ is in the range of the one-way function, and aborts if this is not the case. Otherwise proceed as in $\mathcal{H}_0$.
    
    Note that the only difference between this hybrid and the previous hybrid is if the algorithm in $\mathcal{H}_1$ aborts and the algorithm in $\mathcal{H}_0$ does not. This however implies that both $(\crs_0, y_0)$ and $(\crs_1, y_1)$ are invalid and thus contradicts the statistical soundness of the NIWI.
    \item Hybrid $\mathcal{H}_2$: In this hybrid we compute $c_\OWF \sample \SBSH.\com((\ck_0, \ck_{1}), x')$ where $x'$ is the (inefficiently computed) pre-image of either $y_0$ or $y_1$. Note at least one among $y_0$ and $y_1$ is guaranteed to have a pre-image.
    
    By the statistical hiding of the SBSH commitment, we have that $\mathcal{H}_1 \approx_s\mathcal{H}_2$.
    \item Hybrid $\mathcal{H}_3$: Here we compute the ZAP proof $\pi''$ using the alternative branch, i.e.\ using $x'$ and the randomness of $c_\OWF$ as the witness.
    
    This change is statistically indistinguishable by the statistical indistinguishability of the ZAP argument.
    \item Hybrid $\mathcal{H}_4$: Here we compute $c_\NIZK \sample \SBSH.\com((\ck_0, \ck_{1}), 0)$.
    
    By the statistical indistinguishability of the SBSH commitment we have that $\mathcal{H}_3 \approx_s \mathcal{H}_4$.
    
    \item Hybrid $\mathcal{H}_5$: Here we use $\ket{\psi_1}^{\otimes k(\lambda)}$ to compute the NIZK, instead of $\ket{\psi_0}^{\otimes k(\lambda)}$.
    
    Note that at this point the ouput of the distribution is independent of the NIZK proofs $\pi_0$ and $\pi_1$ and therefore $\mathcal{H}_4 \equiv \mathcal{H}_5$.
    
    \item Hybrid $\mathcal{H}_6$: We revert the change done in $\mathcal{H}_4$.
    \item Hybrid $\mathcal{H}_7$: We revert the change done in $\mathcal{H}_3$.
    \item Hybrid $\mathcal{H}_7$: We revert the change done in $\mathcal{H}_2$.
    \item Hybrid $\mathcal{H}_8$: We revert the change done in $\mathcal{H}_1$.
\end{itemize}
Note that the last hybrid corresponds to the distribution $\ZAPR.\Prove(\crs, \ket{\psi_1}^{\otimes 2k(\lambda)}, x)$, which concludes our proof.
\end{proof}

\section{Constrained PRF for BQP}

We define the notion of constrained pseudorandom function (PRF) and we present a construction for BQP.

\subsection{Definition}

We recall the syntax of constrained PRFs. 
\begin{definition}[Constrained PRF for BQP]
A constrained PRF for BQP $(\cPRF.\Gen,\allowbreak\cPRF.\Eval,\allowbreak\cPRF.\Constrain,\allowbreak\cPRF.\CEval)$ consists of the following efficient algorithms.
\begin{itemize}
\item $\cPRF.\Gen(1^\secp)$: On input the security parameter $1^\secp$, the key generation algorithm outputs a the public parameters $\params$, and master secret key $K$.
\item $\cPRF.\Eval(K, x)$: On input the key $K$ and a string $x$, the evaluation algorithm outputs a string $y$.
\item $\cPRF.\Constrain(K, Q)$: On input the key $K$ and a quantum circuit $Q$ (implementing a BQP language), the constrain algorithm returns a constrained key $K_Q$.
\item $\cPRF.\CEval(\params, K_Q, x)$: On input the public parameters $\params$, a constrained key $K_Q$, and a string $x$, the constrained evaluation algorithm returns a string $y$.
\end{itemize}
\end{definition}
For simplicity, we implicitly assume that the constrained key $K_Q$ also contains a (classical) description of the circuit $Q$.
\begin{definition}[Correctness]
A constrained PRF for BQP $(\cPRF.\Gen,\allowbreak\cPRF.\Eval,\allowbreak\cPRF.\Constrain,\allowbreak\cPRF.\CEval)$ is correct if for all negligible functions $\nu$, there exists a negligible function $\mu$ such that for any $\secp \in \mathbb{N}, x \in \{0,1\}^\secp$, and any quantum circuit $Q$ on $\ell$ input bits such that $\Pr[Q(x) = 1] = 1-\nu(\secp)$, it holds that
    $$\Pr\left[\cPRF.\Eval(K, x) =  \cPRF.\CEval(\params, K_Q, x)\right] = 1 - \mu(\secp)$$
    where $(\params,K) \sample \cPRF.\Gen(1^\secp)$ and $K_Q \sample \cPRF.\Constrain(K, Q)$.
\end{definition}
We define the notion of security of interest. In this work, we consider selective security, meaning that the adversary is bound to choose the challenge point ahead of time. We also consider the collusion resistant variant, where the adversary can obtain an unbounded number of constrained keys for circuits of his choice.
\begin{definition}[Security]
A constrained PRF for BQP $(\cPRF.\Gen,\allowbreak\cPRF.\Eval,\allowbreak\cPRF.\Constrain,\allowbreak\cPRF.\CEval)$ is secure if there exists a pair of negligible functions $\nu$ and $\mu$ such that for all $\lambda \in \mathbb{N}$, all points $x^\ast$, and all admissible non-uniform
 QPT distinguishers with quantum advice 
 $\adv = \{\adv_\lambda, \rho_\lambda\}_{\lambda \in \mathbb{N}}$, it holds that
 $$\begin{array}{r}
 \Pr\left[b =\cA_\secp(\params, y^\ast; \rho_\secp)^{\cPRF.\Constrain(K,\cdot)} : \begin{array}{l}(\params,K) \sample \cPRF.\Gen(1^\secp) \\ b \sample \{0,1\} \\ 
 y^\ast = \cPRF.\Eval(K, x^\ast) \text{ if } b=0\\
 y^\ast \sample \{0,1\}^\lambda \text{ if } b=1\\
 \end{array}\right] = 1/2 + \mu(\secp)\end{array}$$
where $\adv$ is admissible if each query $Q$ to $\cPRF.\Constrain(K,\cdot)$ is such that $\Pr[Q(x^\ast) = 1] \leq \nu(\secp)$.
\end{definition}

\subsection{Construction}

We show in the following a construction of a (collusion-resistant) constrained PRF for BQP. We assume the  existence of a puncturable PRF $(\PRF.\Gen, \allowbreak\PRF.\Puncture, \allowbreak\PRF.\Eval)$, an indistinguishability obfuscator $\obf$ for classical circuits, and a key-policy ABE $(\ABE.\Gen,\allowbreak\ABE.\Enc,\allowbreak \ABE.\KeyGen,\allowbreak \ABE.\Dec)$ for BQP, with classical ciphertexts and classical keys. Our scheme $(\cPRF.\Gen,\allowbreak \cPRF.\Eval,\allowbreak\cPRF.\Constrain,\allowbreak \cPRF.\CEval)$ is described in Figure~\ref{fig:qprf}.

\protocol{Constrained PRF for BQP}{A constrained PRF for BQP}{fig:qprf}
{
\begin{itemize}
\item $\cPRF.\Gen(1^\secp)$:
\begin{itemize}
\item Sample two keys $(k,\title{k}) \sample \PRF.\Gen(1^\lambda)$.
\item Sample a pair $(\msk, \mpk) \sample \ABE.\Gen(1^\secp, 1^\secp)$.
\item Compute the obfuscation $\widetilde{\textbf{P}} \sample \obf(1^\lambda, \textbf{P})$ where $\textbf{P}$ is the circuit that, on input some string $x\in\{0,1\}^\secp$, returns
$$
\ABE.\Enc(\mpk, x, \PRF.\Eval(k, x); \PRF.\Eval(\tilde{k}, x)).
$$
The circuit $\textbf{P}[k]$ is padded to the maximum size of $\textbf{P}^\ast$ (defined in the proof of Theorem~\ref{thm:prfsec}).
\item Return the key $K = (k,\msk)$ and $\params= \widetilde{\textbf{P}}$.
\end{itemize}
\item $\cPRF.\Eval(K, x)$:
\begin{itemize}
\item Return $\PRF.\Eval(k, x)$.
\end{itemize}
\item $\cPRF.\Constrain(K, Q)$:
\begin{itemize}
\item Return $\ABE.\KeyGen(\msk, Q)$.
\end{itemize}
\item $\cPRF.\CEval(\params, K_Q, x)$:
\begin{itemize}
\item Compute $c = \widetilde{\textbf{P}}(x)$.
\item Return $\ABE.\Dec(K_Q, c)$.
\end{itemize}
\end{itemize}
}

\paragraph{Correctness.} The statistical correctness of the scheme follows from the correctness of the underlying cryptographic building blocks. We stress that we assume that the ABE scheme has classical ciphertext, and therefore obfuscation for classical circuits suffices.

\paragraph{Security.} In the following we show that the scheme satisfies selective security. One can lift the analysis to the more realistic adaptive settings by standard complexity leveraging.
\begin{theorem}[Security]\label{thm:prfsec}
Let $(\PRF.\Gen,\allowbreak \PRF.\Puncture, \PRF.\Eval)$ be a puncturable PRF, let $\obf$ be an indistinguishability obfuscator for classical circuits, and let $(\ABE.\Gen,\allowbreak\ABE.\Enc,\allowbreak \ABE.\KeyGen,\allowbreak \ABE.\Dec)$ be an ABE for BQP. Then the scheme in Figure~\ref{fig:qprf} is secure.
\end{theorem}

\begin{proof}
The proof proceeds by defining the following series of hybrid and arguing about the indistinguishability of neighbouring experiments.
\begin{itemize}
    \item Hybrid $\mathcal{H}_{0}$: This is the experiment with the challenge bit fixed to $0$, i.e.\ the challenge string is computed as $\cPRF.\Eval(K, x^\ast)$.
    \item Hybrid $\mathcal{H}_{1}$: In this hybrid we replace the key $\tilde{k}$ with $\tilde{k}_{x^\ast} \sample \PRF.\Puncture(\tilde{k}, x^\ast)$ and we define the circuit $\textbf{P}_1$ to return
    $$
    \ABE.\Enc(\mpk, x, \PRF.\Eval(k, x); \PRF.\Eval(\tilde{k}_{x^\ast}, x))
    $$
    if $x \neq x^\ast$, and $\ABE.\Enc(\mpk, x^\ast, \PRF.\Eval(k, x^\ast); \PRF.\Eval(\tilde{k}, x^\ast))$ otherwise.
    
    Note that $\textbf{P}$ and $\textbf{P}_1$ are functionally equivalent and therefore this modification is computationally indistinguishable by the security of the obfuscator $\obf$.
    
    \item Hybrid $\mathcal{H}_{2}$: In this hybrid we replace the key ${k}$ with ${k}$ with ${k}_{x^\ast} \sample \PRF.\Puncture({k}, x^\ast)$ and we define the circuit $\textbf{P}_2$ to return
    $$
    \ABE.\Enc(\mpk, x, \PRF.\Eval(k_{x^\ast}, x); \PRF.\Eval(\tilde{k}_{x^\ast}, x))
    $$
    if $x \neq x^\ast$, and $\ABE.\Enc(\mpk, x^\ast, \PRF.\Eval(k, x^\ast); \PRF.\Eval(\tilde{k}, x^\ast))$ otherwise.
    
    Note that $\textbf{P}_1$ and $\textbf{P}_2$ are functionally equivalent and therefore this modification is computationally indistinguishable by the security of the obfuscator $\obf$.
    
    \item Hybrid $\mathcal{H}_{3}$: In this hybrid we define the circuit $\textbf{P}_3$ to be identical to $\textbf{P}_2$ except that on input $x^\ast$ it returns  
    $$
    \ABE.\Enc(\mpk, x^\ast, \PRF.\Eval(k, x^\ast); \tilde{r})
    $$
    where $\tilde{r} \sample \{0,1\}^\lambda$.
    
    By the security of the puncturable PRF, this modification is computationally indistinguishable.
    
    \item Hybrid $\mathcal{H}_{4}$: In this hybrid we define the circuit $\textbf{P}^\ast$ to be identical to $\textbf{P}_3$ except that on input $x^\ast$ it returns
    $$
    \ABE.\Enc(\mpk, x^\ast, 0; \tilde{r}).
    $$
    
    This change is computationally indistinguishable by the (selective) security of the ABE scheme.
    
    \item Hybrid $\mathcal{H}_{5}$: In this hybrid we sample challenge string uniformly from $\{0,1\}^\lambda$.
    
    This modification is computationally indistinguishable by the security of the puncturable PRF.
    
    \item Hybrid $\mathcal{H}_{6}$: We revert the change done in $\mathcal{H}_4$.
    \item Hybrid $\mathcal{H}_{7}$: We revert the change done in $\mathcal{H}_3$.
    \item Hybrid $\mathcal{H}_{8}$: We revert the change done in $\mathcal{H}_2$.
    \item Hybrid $\mathcal{H}_{9}$: We revert the change done in $\mathcal{H}_1$.
\end{itemize}
The proof is concluded by observing that $\mathcal{H}_9$ is identical to the experiment where the challenge bit is fixed to $1$.
\end{proof}

\section{Secret Sharing for Monotone QMA}

In the following we show how a witness encryption scheme for QMA directly implies a secret sharing scheme for monotone QMA ($\mQMA$). A language $\lang = ({\lang}_{\text{yes}}, {\lang}_{\text{no}})$ is in monotone QMA if $\lang \in \QMA$ and i) for all statements $x \in {\lang}_{\text{yes}}$ and $y$ such that $x \subseteq y$ it holds that $y \in {\lang}_{\text{yes}}$, and ii) for all statements $x \in \lang_\no$ and $y$ such that $y \subseteq x$ it holds that $y \in \lang_\no$. Here, by $x \subseteq y$ for binary strings $x,y$, we mean that for each index $i$, if $x_i = 1$ then $y_i=1$.

\subsection{Definition}

We begin by defining the notion of secret sharing for languages monotone QMA. For convenience, we only define the scheme for binary secrets, although it is easy to extend it to arbitrary length strings.

\begin{definition}[Secret Sharing]
A secret sharing scheme $(\share, \reconstruct)$ for $N$ parties and a language $\lang\in\mQMA$ consists of the following efficient algorithms.
\begin{itemize}
    \item $\share(1^\lambda, s)$: On input the security parameter $1^\lambda$ and a secret $s\in\{0,1\}$, the sharing algorithm returns a set of $N$ shares $(p_1, \dots, p_N)$.
    \item $\reconstruct(p_1, \dots, p_{|I|}, \ket{\psi}^{\otimes k(\lambda)})$: On input a set $I$ of shares $(p_1, \dots, p_{|I|})$ and a $k(\lambda)$ copies of a quantum state $\ket{\psi}$, the reconstruction algorithm returns a secret $s$.
\end{itemize}
\end{definition}
Next we define the notion of correctness. We say that a set of parties $I$ is qualified if its binary representation defines a statement $x\in\lang_\yes$.
\begin{definition}[Correctness]
A secret sharing scheme $(\share, \reconstruct)$ is correct if there exists a negligible function $\nu$ such that for all $\lambda\in\mathbb{N}$, all secrets $s\in\{0,1\}$, all qualified sets of parties $I$ that define a statement $x \in \lang_\yes$, and all $\ket{\psi}\in\relation(x)$, it holds that
\[
\Pr\left[s= \reconstruct(p_1, \dots, p_{|I|}, \ket{\psi}^{\otimes k(\lambda)})\right] = 1 -\nu(\lambda)
\]
where $(p_1, \dots, p_N) \sample \share(1^\lambda, s)$.
\end{definition}
Finally we define the notion of (non-uniform) computational security of a secret sharing scheme. In the following we say that a set $S$ is unauthorized if it defines a statement $x\in\lang_\no$.
\begin{definition}[Security]
A secret sharing scheme $(\share, \reconstruct)$ is secure if for all $\lambda\in\mathbb{N}$ and all unauthorized sets $S \subseteq \{1, \dots, N\}$,  it holds that
\[
\{p_{i,0}\}_{i\in S} \approx_c \{p_{i,1}\}_{i\in S}
\]
where $(p_{1,0}, \dots, p_{N,0}) \sample \share(1^\lambda, 0)$ and $(p_{1,1}, \dots, p_{N,1}) \sample \share(1^\lambda, 1)$.
\end{definition}

\subsection{Perfectly Binding Commitments}

A (non-interactive) commitment scheme $\com$ is a PPT algorithm that takes as input a message $m$ and returns a commitment $c$. We require that it satisfies the notions of perfect binding and computational hiding, which we define below.

\begin{definition}[Perfect Binding]
A commitment scheme $\com$ is perfectly binding if for all $(m_0, m_1)$ such that $m_0 \neq m_1$ and all $(r_0, r_1)\in \{0,1\}^{2\lambda}$ it holds that $\com(m_0; r_0) \neq \com(m_1; r_1)$.
\end{definition}

\begin{definition}[Computational Hiding]
A commitment scheme $\com$ is computationally hiding if for all $\lambda \in \mathbb{N}$ and all $(m_0, m_1)$ it holds that
\[
\com(m_0; r_0) \approx_c \com(m_1; r_1)
\]
where $(r_0, r_1) \sample \{0,1\}^{2\lambda}$.
\end{definition}

\subsection{Construction}

We describe our secret sharing scheme $(\share, \reconstruct)$ for any language $\lang = ({\lang}_{\text{yes}}, {\lang}_{\text{no}})\in\mQMA$ in the following. Let $\com$ be a perfectly binding commitment scheme and let $(\WE.\Enc, \WE.\Dec)$ be a witness encryption for the language 
$\tilde{\lang} = (\tilde{\lang}_{\text{yes}}, \tilde{\lang}_{\text{no}})$ defined as
\begin{align*}
&\tilde{\lang}_{\text{yes}} = \left\{ (c_1, \dots, c_N) : \begin{array}{c} \exists \text{ a vector }(r_1, \dots, r_N) \in \{0,1\}^{N\lambda} \text{ such that } x \in {\lang}_{\text{yes}}\\ \text{where } x_i = 1 \text{ if } c_i = \com(i; r_i) \text{ and } x_i = 0  \text{ otherwise}.\end{array}\right\} \\
&\tilde{\lang}_{\text{no}} = \left\{ (c_1, \dots, c_N) : \begin{array}{c} \forall \text{ vectors }(r_1, \dots, r_N) \in \{0,1\}^{N\lambda} \text{ it holds that } x \in {\lang}_{\text{no}}\\ \text{where } x_i = 1 \text{ if } c_i = \com(i; r_i) \text{ and } x_i = 0  \text{ otherwise}.\end{array}\right\}
\end{align*}
Our scheme $(\share, \reconstruct)$ is shown in Figure~\ref{fig:secshare}.

\protocol{Secret Sharing for mQMA}{A secret sharing scheme for (monotone) QMA}{fig:secshare}
{
\begin{itemize}
    \item $\share(1^\lambda, s)$:
    \begin{itemize}
        \item For all $i \in [1, \dots, N]$ sample $r_i \sample \{0,1\}^\lambda$ and compute $c_i = \com(i; r_i)$.
        \item Compute $c \sample \WE.\Enc(1^\lambda, (c_1, \dots, c_N), s)$.
        \item Set the share of the $i$-th party to $p_i = (r_i, c)$.
    \end{itemize}
    \item $\reconstruct(p_1, \dots, p_{|I|}, \ket{\psi}^{\otimes k(\lambda)})$: 
    \begin{itemize}
    \item A subset $I \subseteq \{1, \dots, N\}$ of parties parses their shares as $\{p_i = (r_i, c)\}_{i\in I}$.
    \item For all $i \notin I$ set $r_i = \bot$.
    \item Return $\WE.\Dec(c, (r_1, \dots, r_n, \ket{\psi}^{\otimes k(\lambda)}), (c_1, \dots, c_N))$.
    \end{itemize}
\end{itemize}
}

\paragraph{Correctness.} For the correctness of the scheme, observe that any authorized set of users always has a valid witness for the statement $(c_1, \dots, c_N)$ and therefore, by the correctness of the witness encryption scheme, the reconstruction procedure returns the secret $s$, except with negligible probability. 

\paragraph{Security.} In the following we argue about the (non-uniform) security of the scheme.

\begin{theorem}[Security]
Let $\com$ be a computationally hiding commitment scheme and let $(\WE.\Enc, \WE.\Dec)$ be a witness encryption scheme for $\tilde{\lang}$. Then the scheme in Figure~\ref{fig:secshare} is secure.
\end{theorem}
\begin{proof}
To prove this claim, we assume towards contradiction that there exists an efficient quantum algorithm that is able to distinguish between the shares of $0$ and the shares of $1$ given the shares of some non-authorized set of users $S$ with probability $1/2 + \delta(\lambda)$, for some non-negligible function $\delta$. We derive a reduction against the computational hiding property of the commitment $\com$. Specifically, we show that such an algorithm implies a distinguisher between the following distributions
\[
(\com(1), \dots, \com(N)) \text{ and } (\com(\bot), \dots, \com(\bot))
\]
where $\bot$ is some distinguished string that does not correspond to any index. This implies a contradiction to the hiding property of the commitment scheme, by a standard hybrid argument.

On input the challenge set of commitments $(c_1, \dots, c_N)$, the reduction defines a new set of commitments $(c_1', \dots,  c_N')$ by setting $c_i' = c_i$ if $i \notin S$ and $c_i = \com(i;r_i)$ if $i \in S$, where $r_i \sample \{0,1\}^\lambda$. Finally, it samples a bit $b \sample \{0,1\}$ and computes $c \sample \WE.\Enc(1^\lambda, (c_1', \dots,  c_N'), b)$. The distinguisher is given the shares $\{r_i, c\}_{i\in S}$ and returns a bit $b'$. The reduction returns $1$ if $b = b'$ and $0$ otherwise.

In the first case, i.e.\ $(c_1, \dots, c_N) = (\com(1), \dots, \com(N))$ then the distribution of the shares given to the ditinguisher is identical to the one output by the algorithm $\share$ and therefore the probability that the reduction outputs $1$ is identical to $1/2  + \delta(\lambda)$. In the second case, where $(c_1, \dots, c_N) = (\com(\bot), \dots, \com(\bot))$, observe that the resulting instance $(c_1', \dots, c_N') \in \tilde{\lang}_{\text{no}}$ since the commitment scheme is perfectly binding and the set $S$ is non-authorized, i.e.\ it holds that $x\in\lang_\no$ (and consequently that $z\in\lang_\no$, for all $z \subseteq x$), where $x$ is the statement defined by $S$. By the semantic security of the witness encryption scheme it holds that the probability the distinguisher correctly guesses the bit $b$ (and consequently the reduction outputs $1$) is negligibly close to $1/2$. It follows that the reduction successfully distinguishes between the two distributions of commitments with non-negligible probability, which is a contradiction to the hiding property of the commitment scheme.
\end{proof}

\fi

\section{Cryptanalysis of Quantum Obfuscation Candidates}\label{sec:cryptanalysis}

In the following we argue that a natural extension of our approach, where the obfuscated CVQC verification circuit can be queried on accepting instances, leads to insecure schemes, given current constructions of CVQC.
Consider the two-message CVQC of  \cite{FOCS:Mahadev18a,TCC:ChiChuYam20,TCC:ACGH20}. The first message $\pk_1,\dots,\pk_k$ essentially commits to a string $x \in \{0,1\}^k$ where each bit $x_i$ determines whether the $i$'th qubit of the prover's committed state will be measured in the computational basis or Hadamard basis. It is crucial that $x$ be hidden from the prover, since a prover that can predict how its state will be measured can easily break soundness of the underlying information-theoretic protocol of \cite{Fitzsimons_2018}. Now, the prover's proof includes a series of values $(b_1,d_1),\dots,(b_k,d_k)$, where each $b_i$ is a bit and each $d_i$ is a string.\footnote{Actually some positions will be designated for a test round, in which case the proof at those positions will have a different structure, but this can be ignored for the purpose of this discussion.} For each position where $x_i = 0$ (indicating a computational basis measurement), the verifier will completely ignore $(b_i,d_i)$, and for each position where $x_i = 1$ (indicating a Hadamard basis measurement), the verifier will compute some bit $e_i \coloneqq b_i \oplus (d_i \cdot s_i)$, where $s_i$ is some secret known to the verifier. The $e_i$'s are then computed on as part of the verifier's verdict function. 

Consider an adversary, with access to an obfuscation of the verifier, that first honestly computes an accepting proof that includes values $(b_1,d_1),\dots,(b_k,d_k)$. Then, it generates random $b'_1,\dots,b'_k$ until the same proof except with $(b'_1,d_1),\dots,\allowbreak (b'_k,d_k)$ rejects. Now the adversary can flip $b'_i$ to $b_i$ one at a time until the proof accepts again. If the proof flipped back to accepting at index $i$, then it must be the case that $x_i = 1$, since $(b_i,d_i)$ was not ignored by the verifier. An adversary can repeat this process until it learns enough of $x$ to break soundness. Once the adversary can make the verifier accept even on rejecting instances, it is able to learn any other secrets hidden in the obfuscation of the verifier (concretely, information about the obfuscated quantum circuit).

The above describes an attack against the basic scheme, but there may be ways of altering the protocol to avoid such an attack. Below, we argue that three natural variants of the protocol will still be susceptible to attacks. 

\paragraph{Sampling Fresh Hamiltonian Terms.} The above attack does assume that the same verification function is applied to the $e_i$'s across all of the different proofs that the prover submits. This will be the case if the verifier has sampled and fixed the same set of Hamiltonian terms to measure. However, one could consider re-sampling the Hamiltonian terms each time, by using randomness derived from applying a PRF to the proof. Then, when the adversary observes that flipping $b'_i$ to $b_i$ results in the verification flipping from rejecting to accepting, it may not be clear whether this flip occurred as a result of changing $e_i$, or as a result of sampling a new subset of indices to compute on (corresponding to a different set of Hamiltonian terms). Unfortunately, this does not solve the issue that the keys $\pk_1,\dots,\pk_k$ commit to a fixed string $x$ of measurement bases. Indeed, an adversary could still compute accepting/rejecting statistics for $b_i$ vs. $b_i'$, by altering other parts of the proof. If the acceptance probability remains roughly the same, it is likely that $x_i = 0$, and otherwise it is likely that $x_i = 1$.

\paragraph{Removing the Commitment.} One could try to design a CVQC protocol where $\pk_1,\dots,\pk_k$ does not commit to $x$, i.e.\ each $\pk_i$ is the public key for a trapdoor claw-free function family (as opposed to Mahadev's protocol~\cite{FOCS:Mahadev18a}, where some of the public keys are for injective functions). However, the ability to flip the verifier's verdict based on just flipping $(b_i,d_i)$ to $(b_i',d_i)$ will still lead to an attack. Indeed, the adversary can now alter the $d_i$ in order to learn enough linear equations of the verifier's secret $s_i$ to eventually recover $s_i$. The value $s_i$ is actually related to the LWE secret underlying the $\pk_i$, allowing the adversary to break the claw-freeness of the function, and thus soundness of the protocol.

\paragraph{Quantum CVQC.} Recently, \cite{morimae2021classically} showed how to construct an information-theoretically secure two-message verification protocol where the first message from verifier to prover is quantum, but the proof and subsequent verification function are entirely classical. Thus, one could still use classical obfuscation to obfuscate the verification circuit (though now the obfuscation of the quantum circuit would be a quantum state). However, this scheme would succumb to the same class of attacks outlined above. Indeed, the protocol of \cite{morimae2021classically} roughly works by having the verifier prepare several EPR pairs and release half of each to the prover. The prover then prepares a (quantum) proof for \cite{Fitzsimons_2018}'s protocol and \emph{teleports} this proof into the verifier's halves of the EPR pairs. That is, the prover performs Bell measurements between its proof and the quantum state sent by the verifier, and the final proof just consists of the resulting teleportation errors. Now, the verification can be made completely classical by \emph{pre-measuring} (in either the computational or Hadamard basis) each of the verifier halves of the EPR pairs and hard-coding these measurement results into the verification circuit. Unfortunately, this also implies that qubit $i$ sent by the verifier is again essentially a commitment to measuring qubit $i$ of the prover's proof in either the computational or Hadamard basis, and the attack sketched above will still apply.

\subsection*{Acknowledgements}
The authors wish to thank Dakshita Khurana for many insightful discussions.

\bibliographystyle{alpha}
\bibliography{abbrev3,crypto,extra}

\appendix

\ifsubmission
\pagebreak
\begin{center}
    {\huge\bfseries SUPPLEMENTARY MATERIAL}
\end{center}

\section{An Alternative Obfuscation of Null Quantum Circuits}\label{sec:vbb}
In the following we present an alternative construction of null-iO for quantum circuits that uses classical VBB obfuscation. 

\paragraph{Classical VBB Obfuscation.}
First, we recall the definition of \emph{virtual black-box obfuscation}.
\begin{definition}[Virtual Black-Box Obfuscation]\label{def:VBB}
For every PPT adversary $\cA$, there exists a PPT simulator $\Sim$ such that for all circuits $C$, and all polynomial-size auxiliary input $z$, it holds that $$\bigg|\Pr\left[\cA\left(\obf(1^\secp,C),z\right) = 1\right] = \Pr\left[\Sim^C\left(1^\secp,1^{|C|},z\right) = 1\right]\bigg| \leq \negl.$$
\end{definition}
This definition extends to the post-quantum setting by allowing the adversary and simulator to be QPT, and giving the simulator superposition query access to $C$.

\paragraph{Blind CVQC.} Next, we define and construct \emph{blind} CVQC.

\begin{definition}[Blindness]
    A CVQC protocol $(\KeyGen, \Prove, \Verify)$ is blind if for any QPT adversary $\cA$ there exists a negligible function $\nu(\secp)$  such that for any polynomial-size sequences of circuits $\{Q_{0,\secp}\}_{\secp \in \bbN}, \{Q_{1,\secp}\}_{\secp \in \bbN}$, it holds that
    $$\begin{array}{r}\left|\Pr\left[\cA^{\ket{\cH}}(\params) = 1 : (\params,r) \sample \KeyGen(1^\secp,Q_{0,\secp})\right] - \Pr\left[ \cA^{\ket{\cH}}(\params) = 1: (\params,r) \sample \KeyGen(1^\secp,Q_{1,\secp})\right]\right| = \nu(\secp).\end{array}$$
\end{definition}

We observe that blind CVQC exists in the QROM from LWE. While the works of~\cite{TCC:ChiChuYam20,TCC:ACGH20} give CVQC protocols that satisfy correctness and soundness, their protocols are not blind. However, using quantum fully-homomorphic encryption (Definition~\ref{def:qfhe}), it is straightforward to obtain a two-message blind protocol in the QROM, by applying Fiat-Shamir to a blind variant of the four-message CVQC protocol from~\cite{TCC:ChiChuYam20,TCC:ACGH20}. 

\begin{lemma}
Assuming the quantum hardness of learning with errors problem, there exists a two-message blind classical verification of quantum computation protocol in the QROM.
\end{lemma}

\begin{proof}
Consider the four-message CVQC protocol from~\cite{TCC:ChiChuYam20,TCC:ACGH20}, which is correct and sound assuming the quantum hardness of learning with errors. The protocol includes two prover algorithms $\Prove_1$ and $\Prove_2$, and operates as follows. First, $\KeyGen(1^\secp,Q)$ is run to produce parameters $\params$ and a verification key $r$. Next, $\Prove_1(\params,\ket{\psi})$ outputs a classical string $y$ and quantum state $\ket{\psi_1}$. Next, a uniformly random classical challenge $c$ is sampled, and $\Prove_2(\params,y,c,\ket{\psi_1})$ is run to produce a classical proof $\pi$. Finally, $\Verify(Q,y,c,\pi,r)$ outputs a bit indicating acceptance of rejection.

This protocol consisting of algorithms $(\KeyGen,\Prove_1,\Prove_2,\Verify)$ can be turned into a blind protocol $(\KeyGen',\Prove'_1,\Prove'_2,\Verify')$ as follows. $\KeyGen'$ will run $\KeyGen$, sample a random QFHE key pair $(\pk,\sk)$, and then output $\Enc(\pk,\params)$. $\Prove'_1$ will now run $\Prove_1$ under the QFHE to produce $(\pk,\Enc(\pk,y))$. The challenge $c$ will still be sampled and given to the prover in the clear. $\Prove'_2$ will now run $\Prove_2$ under the QFHE to produce $\Enc(\pk,\pi)$. Finally, the verification procedure $\Verify'(Q,\Enc(\pk,y),c,\Enc(\pk,\pi),r)$ will compute $\sk$ from $r$, decrypt the ciphertexts to obtain $y$ and $\pi$, and then run $\Verify(Q,y,c,\pi,r)$.

Correctness of $(\KeyGen',\Prove'_1,\Prove'_2,\Verify')$ follows immediately from correctness of QFHE. Soundness follows by a reduction to the soundness of $(\KeyGen,\allowbreak \Prove_1,\Prove_2,\Verify)$, in which the reduction samples the $(\pk,\sk)$ key pair, encrypts $\params$ received from its challenger, and decrypts each of $\Enc(\pk,y)$ and $\Enc(\pk,\pi)$ before forwarding them to its challenger. Blindness follows immediately from the semantic security of QFHE.

Finally, we can compile $(\KeyGen',\Prove'_1,\Prove'_2,\Verify')$ into a two-message protocol in the QROM, by appealing to the ``Fiat-Shamir for generalized $\Sigma$-protocols'' Lemma~\cite[Lemma 6.2]{TCC:ACGH20}.
\end{proof}

\paragraph{Construction.}
Finally, we present the construction. Let $(\KeyGen,\Prove,\Verify)$ be a two-message blind CVQC scheme in the QROM (Definition~\ref{def: classical argument}). Let $\obf$ be a post-quantum virtual black-box obfuscator for classical circuits (Definition~\ref{def:VBB}). Let $F : \{0,1\}^\secp \times \{0,1\}^* \to \{0,1\}^{m(\secp)}$ be a quantum-secure PRF.

\protocol{Alternative Null-iO for Quantum Circuits}{A null obfuscation scheme for quantum circuits from classical VBB obfuscation.}{fig:null-iO}
{
\begin{itemize}
    \item $\obf(1^\secp,Q)$: 
    \begin{itemize}
        \item Sample $(\params, r) \sample \KeyGen(1^\secp,Q)$.
        \item Sample a PRF key $k \sample \{0,1\}^\secp$ and define $C_0(\cdot) = F(k,\cdot)$.
        \item Let $V[Q,k,r](\cdot)$ be the circuit that takes as input $\pi$ and computes and outputs $b = \Verify^{F(k,\cdot)}(Q,\pi,f)$. Define $C_1(\cdot) = V[Q,k,r](\cdot)$.
        \item Define $C(b,x) = C_b(x)$, and compute $\widetilde{C} \sample \obf(1^\secp,C)$.
        \item Output $\widetilde{Q} = (\params,\widetilde{C})$.
    \end{itemize}
    \item $\Eval(\widetilde{Q},\ket{\psi})$: Compute $\pi \sample \Prove^{\widetilde{C}(0,\cdot)}(\params,\ket{\psi})$ and output $\widetilde{C}(1,\pi)$.
\end{itemize}
}

\begin{theorem}
Assuming post-quantum virtual black-box obfuscation of classical circuits and the quantum hardness of learning with errors, the scheme in Figure~\ref{fig:null-iO} is a secure null-iO for quantum circuits.
\end{theorem}

\begin{proof}
Correctness follows immediately from correctness of CVQC and the VBB obfuscator, and security of the PRF.

To show security, fix any two null circuits $Q_0,Q_1$ (these are technically families of circuits, but we drop the indexing by $\secp$ to avoid clutter), and a QPT adversary $\cA$. Recall that virtual black-box obfuscation guarantees the existence of a simulator $\Sim$ based on the description of $\cA$. Let $q$ be an upper bound on the number of queries $\Sim$ makes to its oracle.

\begin{itemize}
    \item $\cH_0$: Adversary $\cA$ receives $\obf(1^\secp,Q_0)$, which consists of $\params$, and a classical obfuscated circuit $\widetilde{C}$.
    \item $\cH_1$: Simulator $\Sim$ receives $\params$ and is given (superposition) oracle access to $C$. The output of this hybrid is indistinguishable from $\cH_0$ by the security of classical VBB obfuscation.
    \item $\cH_2$: Any calls to $F(k,\cdot)$ during the evaluation of $C$ are forwarded to a quantum-accessible random oracle $\cH$. Since $F$ is a quantum-secure PRF, this is indistinguishable from $\cH_1$.
    \item $\cH_{3,i}$ for $i \in [q]$: The simulator's first $i$ queries to $\widetilde{C}$ are answered as follow. $$\sum_x \alpha_{0,x}\ket{0,x}\ket{0} + \sum_\pi \alpha_{1,\pi}\ket{1,\pi}\ket{0} \to \sum_x \alpha_{0,x}\ket{0,x}\ket{\cH(x)} + \sum_\pi \alpha_{1,\pi}\ket{1,\pi}\ket{0}.$$ The indistinguishability $\cH_{3,i} \approx_c \cH_{3,i-1}$ follows from the soundness of CVQC. Indeed, an adversary can only distinguish if its $i$'th query has some inverse polynomial amplitude on $\pi$ such that $\Verify^\cH(Q_0,\pi,r) = 1$. Otherwise, the hybrids would be statistically close. However, in this case, a reduction can produce an accepting proof for CVQC with inverse polynomial probability by answering each of the first $i-1$ queries as in $\cH_{3,i-1}$ and then measuring the $i$'th query. This violates the soundness of CVQC, since $Q_0$ rejects on all inputs with overwhelming probability.
    \item $\cH_4$: Switch the obfuscation to $Q_1$. Indistinguishability follows from the blindness of CVQC. Indeed, the verifier's secret key is no longer needed to simulate this distribution, since the $\Verify$ circuit is no longer invoked.
    \item $\cH_{5,i}$ for $i \in [q]$: reverse the changes from $\cH_{3,i}$ for $i \in [q]$.
    \item $\cH_6$: reverse the change from $\cH_2$. 
    \item $\cH_7$: reverse the change from $\cH_1$. This is $\obf(1^\secp,Q_1)$.
\end{itemize}

\end{proof}

\section{One-Sided Attribute-Hiding}\label{sec:pe}

In the following we show how to convert generically an ABE for BQP into a predicate encryption scheme with one-sided security (i.e.\ where the quantum circuit associated with a ciphertext is hidden to the eyes of an adversary that cannot decrypt), additionally assuming the quantum hardness of the LWE problem. Towards this goal, we introduce the notion of quantum lockable obfuscation and we present a scheme from quantum-hard LWE.

\subsection{Definition}

We define the notion of one-sided attribute-hiding for attribute based encryption. In the literature, this primitive is also referred to as predicate encryption~\cite{C:GorVaiWee15}. Since the syntax is unchanged, we only present the upgraded security definition. Similarly as before, we consider the case of selective security, where the quantum circuit associated with the challenge ciphertext is fixed ahead of time, and in particular is chosen before seeing the public parameters of the scheme.
\begin{definition}[One-Sided Attribute-Hiding]
An ABE scheme $(\ABE.\Gen,\allowbreak\ABE.\Enc,\ABE.\KeyGen,\allowbreak\ABE.\Dec)$ is one-sided attribute-hiding if there exists a triple of negligible functions $\nu_0$, $\nu_1$, and $\mu$ such that for all $\lambda \in \mathbb{N}$, all pairs of quantum circuits $(Q_0, Q_1)$, and all admissible non-uniform
 QPT distinguishers with quantum advice 
 $\adv = \{\adv_\lambda, \rho_\lambda\}_{\lambda \in \mathbb{N}}$, it holds that
 $$\Pr\left[b =\cA_\secp(\ct_{Q_b}, \mpk; \rho_\secp)^{\ABE.\KeyGen(\msk,\cdot)} : \begin{array}{l}(\mpk,\msk) \sample \ABE.\Gen(1^\secp,1^\ell) \\ b \sample \{0,1\} \\ \ct_{Q_b} \sample \ABE.\Enc(\mpk,Q_b,b)\end{array}\right] = 1/2 + \mu(\secp),$$
where $\adv$ is admissible if each query $x$ to $\ABE.\KeyGen(\msk,\cdot)$ is such that $\Pr[Q_0(x) = 1] = \nu_0(\secp)$ and $\Pr[Q_1(x) = 1] = \nu_1(\secp)$.
\end{definition}

\subsection{Quantum Lockable Obfuscation}

In the following we give a construction of quantum lockable obfuscation assuming the quantum hardness of the LWE problem.

\paragraph{Compute-and-Compare Programs.} We define the class of compute-and-compare circuits. The definition below applies both to classical and pseudo-deterministic quantum circuits (with classical input and output).

\begin{definition}[Compute-and-Compare]
Let $C:\{0,1\}^n\rightarrow\{0,1\}^\lambda$ be a circuit, and let $u\in\{0,1\}^\lambda$ and $z \in \{0,1\}^\ast$ be two classical strings. Then $\CC{C}{u}{z}(x)$ is a circuit that returns $z$ if $C(x)=u$, and $0$ otherwise.  
\end{definition}

\paragraph{Definition.} We are now ready to define the notion of lockable obfuscation for compute-and-compare programs. In what follows we only define the classical version of lockable obfuscation. The extension to quantum circuits follows along the same lines. A lockable obfuscator $\obf$ is a PPT algorithm that takes as input a compute-and-compare program $\CC{C}{u}{z}$ and outputs a new circuit $\obfC$. We assume that the circuit $\CC{C}{u}{z}$ is given in some canonical description from which $C$, $u$, and $z$ can be read. Correctness is defined as follows.

\begin{definition}[Correctness]
A lockable obfuscator $\obf$ is correct if there exists a negligible function $\nu$ such that for all $\lambda\in\mathbb{N}$, all circuits $C:\{0,1\}^n\rightarrow\{0,1\}^\lambda$, all $u\in\{0,1\}^\lambda$, and all $z \in \{0,1\}^\ast$, it holds that
\[
\Pr\left[\forall x \in \{0,1\}^n: \obfC(x)= \CC{C}{u}{z}(x)\right] = 1 - \nu(\secp)
\]
where $\obfC\sample\obf(1^\lambda, \CC{C}{u}{z})$.
\end{definition}
We require a strong notion of simulation security for lockable obfuscation.
\begin{definition}[Simulation Security]
A lockable obfuscator $\obf$ is secure if there exists a simulator $\ccSim$ such that for all $\lambda\in\mathbb{N}$, all pseudo-deterministic polynomial-size quantum circuits $C:\{0,1\}^n\rightarrow\{0,1\}^\lambda$, and all polynomial-length output strings $z \in \{0,1\}^\ast$ it holds that
\[
\obf(1^\lambda, \CC{C}{u}{z}) \approx_c \ccSim(1^\lambda, 1^{|C|}, 1^{|z|})
\]
where $u \sample \{0,1\}^\lambda$.
\end{definition}

\paragraph{Classical Lockable Obfuscation.} Lockable obfuscation for classical circuits and with almost perfect correctness were constructed in \cite{FOCS:GoyKopWat17,FOCS:WicZir17}, assuming the quantum hardness of LWE. Recently, a construction with perfect correctness has been shown in~\cite{TCC:GKVW20}.

\begin{lemma}[\cite{FOCS:GoyKopWat17,FOCS:WicZir17}]
Assuming the quantum hardness of the LWE problem, there exists a lockable obfuscation $\obf$ for compute-and-compare programs.
\end{lemma}

\paragraph{Quantum Lockable Obfuscation.} We now propose a lockable obfuscation scheme for pseudo-deterministic compute-and-compare \emph{quantum programs}. The scheme combines a QFHE scheme $(\QFHE.\Gen,\allowbreak\QFHE.\Enc,\allowbreak \QFHE.\Eval,\allowbreak \QFHE.\Dec)$ with classical keys and classical decryption with a lockable obfuscator $\obf$ for classical circuits. The scheme $\obf^\ast$ is described in Figure~\ref{fig:qlock}. 

\protocol{Quantum Lockable Obfuscation}{A lockable obfuscator for (pseudo-deterministic) compute-and-compare quantum circuits}{fig:qlock}
{
\begin{itemize}
    \item $\obf^\ast(1^\lambda, \CC{Q}{u}{z})$:
    \begin{itemize}
        \item Sample $(\sk, \pk) \sample \QFHE.\Gen(1^\lambda)$.
        \item Compute the classical encryption $c \sample \QFHE.\Enc(\pk, Q)$ of the circuit $Q$.
        \item Compute the obfuscation $\obfC \sample \obf(1^\lambda, \CC{\QFHE.\Dec(\sk, \cdot)}{u}{z})$ of the classical compute-and-compare program that takes as input a ciphertext $c^\ast$ and returns $z$ if and only if $\QFHE.\Dec(\sk, c^\ast) = u$.
        \item Return $(c, \obfC)$.
    \end{itemize}
\end{itemize}
To evaluate the obfuscated circuit $(c, \obfC)$ on some input $x$, the (quantum) algorithm computes $\tilde{c} \sample \QFHE.\Eval(\pk, \mathcal{U}_x, c)$ as the homomorphic evaluation of the universal quantum circuit $\mathcal{U}_x$ that has hard-wired input $x$ and evaluates the encrypted circuit $Q$ homomorphically. Then return the evaluation of $\obfC$ on input the resulting $\tilde{c}$.
}

Correctness follows straightforwardly from the correctness of the underlying building blocks. One caveat is that the QFHE evaluation procedure is inherently probabilistic, so correctness only holds with probability negligibly close to one (over the randomness imposed by the QFHE evaluation). We now proceed to establish the simulation security of the scheme.

\begin{theorem}[Simulation Security]
Let $(\QFHE.\Gen,\allowbreak\QFHE.\Enc,\allowbreak \QFHE.\Eval,\allowbreak \QFHE.\Dec)$ be a secure QFHE scheme and let $\obf$ be a simulation secure lockable obfuscator for classical compute-and-compare programs. Then the scheme in Figure~\ref{fig:qlock} is simulation secure.
\end{theorem}
\begin{proof}
To prove the security of our construction, we define a series of hybrid distributions that we argue to be computationally close. The real distribution for an obfuscated quantum circuit $Q$ consists of
\[
(\QFHE.\Enc(\pk, Q), \obf(1^\lambda, \CC{\QFHE.\Dec(\sk, \cdot)}{u}{z})).
\]
In the first hybrid distribution we simulate the obfuscated classical circuit. Since $u$ is uniformly sampled (and in particular is independent of all other variables of our distribution), by an invocation of the simulation security of $\obf$, we obtain that
\begin{align*}
&(\QFHE.\Enc(\pk, Q), \obf(1^\lambda, \CC{\QFHE.\Dec(\sk, \cdot)}{u}{z})) \\&\approx_c (\QFHE.\Enc(\pk, Q), \ccSim(1^\lambda, 1^{|\QFHE.\Dec|}, 1^{|z|})).
\end{align*}
Next, we switch the QFHE to be an encryption of $0^{|U|}$. By the semantic security of the QFHE scheme we have that
\begin{align*}
&(\QFHE.\Enc(\pk, Q), \ccSim(1^\lambda, 1^{|\QFHE.\Dec|}, 1^{|z|})) \\&\approx_c (\QFHE.\Enc(\pk, 0^{|Q|}), \ccSim(1^\lambda, 1^{|\QFHE.\Dec|}, 1^{|z|})).
\end{align*}
Note that the latter distribution does not contain any information about the circuit $Q$, besides its size. This concludes our proof.
\end{proof}

\subsection{Construction}

We now present our transformation to upgrade the security of an ABE scheme for quantum circuit to one-sided attribute-hiding. The transformation is identical to~\cite{FOCS:WicZir17,FOCS:GoyKopWat17} and we describe it here for completeness. We require the existence of an ABE scheme $(\ABE.\Gen,\allowbreak\ABE.\Enc,\allowbreak\ABE.\KeyGen,\ABE.\Dec)$ which (for convenience) we assume that it can encrypt multiple bits, and a lockable obfuscator $\obf$ for compute-and-compare quantum circuits. The scheme $\obf^\ast$ is described in Figure~\ref{fig:pe}. 

\protocol{One-Sided Attribute-Hiding ABE}{A one-sided attribute-hiding ABE scheme}{fig:pe}
{
\begin{itemize}
    \item $\ABE.\Gen^\ast(1^\secp,1^\ell)$:
    \begin{itemize}
        \item Return $\ABE.\Gen(1^\secp,1^\ell)$.
    \end{itemize}
    \item $\ABE.\Enc^\ast(\mpk,Q,m)$:
    \begin{itemize}
        \item Sample a uniform $u \sample \{0,1\}^\lambda$.
        \item Compute $c \sample \ABE.\Enc(\mpk,Q,u)$.
        \item Compute the obfuscation $\obfC \sample \obf(1^\lambda, \CC{\ABE.\Dec(\cdot, c)}{u}{m})$ of the quantum compute-and-compare program that takes as input a key $\sk_x$ and returns $m$ if and only if $\ABE.\Dec(\sk_x, c) = u$.
        \item Return $c_Q =\obfC$.
    \end{itemize}
    \item $\ABE.\KeyGen^\ast(\msk,x)$:
    \begin{itemize}
        \item Return $\ABE.\KeyGen(\msk,x)$.
    \end{itemize}
    \item $\ABE.\Dec^\ast(\sk_x,\ct_Q)$:
    \begin{itemize}
        \item Return the output of the evaluation of $\obfC$ on $\sk_x$.
    \end{itemize}
\end{itemize}
}

\paragraph{Correctness.} The ABE decryption circuit is by definition pseudo-deterministiic and therefore the correctness of the scheme is routinely established by invoking the correctness of the ABE and of the obfuscator for compute-and-compare quantum programs.

\paragraph{Security.} We now analyze the security of the scheme, which we establish with the following theorem.

\begin{theorem}[One-Sided Attribute-Hiding]
Let $(\ABE.\Gen,\allowbreak\ABE.\Enc,\allowbreak\ABE.\KeyGen,\allowbreak\ABE.\Dec)$ be a secure ABE scheme and let $\obf$ be a simulation secure obfuscator for compute-and-compare quantum programs. Then the scheme in Figure~\ref{fig:pe} is one-sided attribute-hiding.
\end{theorem}

\begin{proof}
The proof proceeds by defining the following series of hybrids.
\begin{itemize}
    \item Hybrid $\mathcal{H}_0$: This is the original experiment with the bit fixed to $b=0$.
    
    \item Hybrid $\mathcal{H}_1$: This is identical to the previous hybrid except that in the challenge ciphertext we compute $c \sample \ABE.\Enc(\pk, Q_0, 0^\lambda)$.
    
    Indisinguishability follows from the (standard) security of the ABE scheme.
    
    \item Hybrid $\mathcal{H}_2$: In this hybrid we compute the challenge ciphertext via  the simulator $\ccSim(1^\lambda, 1^{|\ABE.\Dec|}, 1^{|m|}))$.
    
    Note that in $\mathcal{H}_1$ the target value $u$ is uniform to the eyes of the distinguisher and therefore indistinguishbility follows from the simulation security of the compute-and-compare obfuscator.
    
     \item Hybrid $\mathcal{H}_3$: In this hybrid we modify the challenge ciphertext by sampling $c \sample \ABE.\Enc(\pk, Q_1, 0^\lambda)$.
     
     Since $c$ does not affect the view of the adversary this modification is only syntactical.
     
     \item Hybrid $\mathcal{H}_4$: We revert the change done in $\mathcal{H}_2$, except that we compute the obfuscated circuit as $\obfC \sample \obf(1^\lambda, \CC{\ABE.\Dec(\cdot, c)}{u}{m_1})$.
     
     Indistinguishability follows along the same lines as what argued above.

     \item Hybrid $\mathcal{H}_4$: We revert the change done in $\mathcal{H}_1$, except that we compute $c \sample \ABE.\Enc(\pk, Q_1, u)$.
     
     Indistinguishability follows from another invocation of the security of the ABE scheme.
\end{itemize}
The proof is concluded by observing that the distribution induced by $\mathcal{H}_4$ is identical to the original experiment with the bit fixed to $b=1$.
\end{proof}

\else

\fi

\end{document}